\documentclass{lmcs}
\pdfoutput=1

\usepackage{lastpage}
\lmcsdoi{15}{3}{26}
\lmcsheading{}{\pageref{LastPage}}{}{}%
{Jan.~31,~2019}{Sep.~02,~2019}{}

\usepackage{amsmath, amsthm,amsfonts,amssymb,mathtools,mathdots}
\usepackage{hyperref}

\usepackage{enumitem, imakeidx, framed}
\usepackage{tikz, stmaryrd}

\theoremstyle{definition} \newtheorem{axgrp}[thm]{Axioms}

\usetikzlibrary{arrows, decorations.markings, shapes.geometric, decorations.pathmorphing, decorations.pathreplacing, intersections, patterns,calc, backgrounds, matrix,through}

\pgfdeclarelayer{bg}
\pgfdeclarelayer{mid}
\pgfdeclarelayer{fg}
\pgfsetlayers{bg,mid,fg,main}

\tikzset{
	follow/.style={->, >=stealth, very thick, shorten <=3pt, shorten >=3pt, color=magenta!80},
	edge/.style={line width=.8pt, color=black},
	edgedot/.style={densely dotted, line width=.8pt, color=black},
	edgethdot/.style={densely dotted, line width=.8pt, color=gray},
	edgeth/.style={line width=.8pt, color=gray!60},
	mirror/.style={line width=1.6pt, color=magenta!80},
	dot/.style={circle, draw=black, line width=.8pt, fill=white, inner sep=1.7pt},
	dotth/.style={circle, draw=gray!60, fill=gray!60, inner sep=1.25pt},
	dotwh/.style={circle, draw=gray!60, line width=.8pt, fill=white, inner sep=1.8pt},
	dotwhite/.style={circle, draw=gray, line width=.8pt, fill=white, inner sep=1.8pt},
	dotdark/.style={circle, draw, fill=black, inner sep=1.25pt},
	dotgrey/.style={circle, draw=black, line width=.8pt, fill=gray!60, inner sep=1.6pt},
	coverc/.style={circle, draw=magenta!80, line width=.8pt, inner sep=3pt},
	coverch/.style={circle, draw=magenta!30, line width=.8pt, inner sep=3pt},
	every node/.style={scale=.8},
	every picture/.style={scale=.8},
}

\newcommand\cat[1]{\mathbf{#1}}
\newcommand\graded{\cat{Hilb}^{\mathbb{Z}_2}}
\newcommand\cal[1]{\mathcal{#1}}
\newcommand\ftimes{\,\tilde{\otimes}\,}
\newcommand\complex{\mathbb{C}}
\newcommand\fock{\cal{F}}
\newcommand\pset{\cal{P}}

\newcommand\blankline{\vspace{6pt}}

\newcommand\dagg[1]{#1^\dagger}

\newcommand\ket[1]{|\, #1 \, \rangle}

\newcommand{\ropen}[1]{[#1)} 

\begin{document}
\title{A diagrammatic calculus \texorpdfstring{\\}{}of fermionic quantum circuits}

\author[G. de Felice]{Giovanni de Felice\rsuper{a}}
\address{\lsuper{a}Department of Computer Science, University of Oxford}
\email{\{giovanni.defelice,kangfeng.ng\}@cs.ox.ac.uk}

\author[A. Hadzihasanovic]{Amar Hadzihasanovic\rsuper{b}}
\address{\lsuper{b}RIMS, Kyoto University}
\email{ahadziha@kurims.kyoto-u.ac.jp}

\author[K.F. Ng]{Kang Feng Ng\rsuper{a}}

\begin{abstract}
We introduce the fermionic ZW calculus, a string-diagrammatic language for fermionic quantum computing (FQC). After defining a fermionic circuit model, we present the basic components of the calculus, together with their interpretation, and show how the main physical gates of interest in FQC can be represented in our language. We then list our axioms, and derive some additional equations. We prove that the axioms provide a complete equational axiomatisation of the monoidal category whose objects are systems of finitely many local fermionic modes (LFMs), with maps that preserve or reverse the parity of states, and the tensor product as monoidal product. We achieve this through a procedure that rewrites any diagram in a normal form. As an example, we show how the statistics of a fermionic Mach-Zehnder interferometer can be calculated in the diagrammatic language. We conclude by giving a diagrammatic treatment of the dual-rail encoding, a standard method in optical quantum computing used to perform universal quantum computation.
\end{abstract}

\maketitle

\section*{Introduction}

The ZW calculus is a string-diagrammatic language for qubit quantum computing, introduced by the second author in~\cite{hadzihasanovic2015diagrammatic}. Developing ideas of Coecke and Kissinger~\cite{coecke2010compositional}, it refined and extended the earlier ZX calculus~\cite{coecke2008interacting,backens2014zx}, while keeping some of its most convenient properties, such as the ability to handle diagrams as undirected labelled multigraphs. In the version of~\cite[Chapter 5]{hadzihasanovic2017algebra}, it provided the first complete equational axiomatisation of the monoidal category of qubits and linear maps, with the tensor product as monoidal product. Soon after its publication, the third author and Q.\ Wang derived from it a universal completion of the ZX calculus~\cite{ng2017universal,lics2018zwzx}.

Since its early versions, the ZX calculus has had the advantage of including familiar gates from the circuit model of quantum computing~\cite[Chapter 4]{nielsen2009quantum}, such as the Hadamard gate and the CNOT gate, either as basic components of the language, or as simple composite diagrams. This facilitates the transition between formalisms and the application to known algorithms and protocols, and is related to the presence of a simple, well-behaved ``core'' of the ZX calculus, modelling the interaction of two strongly complementary observables~\cite{coecke2012strong}, in the guise of special commutative Frobenius algebras~\cite{coecke2012new}. Access to complementary observables is fundamental in quantum computing schemes such as the one-way quantum computer, to which the ZX calculus was applied in~\cite{duncan2010rewriting}.

The ZW calculus only includes one special commutative Frobenius algebra, corresponding to the computational basis, as a basic component. On the other hand, as noted already in~\cite{hadzihasanovic2015diagrammatic}, the ZW calculus has a fundamentally different ``core'', which is obtained by removing a single component that does not interact as naturally with the rest. This core has the property of only representing maps that have a definite parity with respect to the computational basis: the subspaces spanned by basis states with an even or odd number of 1s are either preserved, or interchanged by a map. This happens to be compatible with an interpretation of the basis states of a single qubit as the empty and occupied states of a \emph{local fermionic mode}, the unit of information of the \emph{fermionic quantum computing} (FQC) model.

Fermionic quantum computing is computationally equivalent to qubit computing~\cite{bravyi2002fermionic}. The connection with the ZW calculus suggested that an independent fermionic version of the calculus could be developed, combining the best of both worlds with respect to FQC rather than qubit computing: the superior structural properties of the ZW calculus, including an intuitive normalisation procedure for diagrams, together with the superior hands-on features of the ZX calculus.

In this paper, we present such an axiomatisation, to which we refer as the \emph{fermionic ZW calculus}. In Section~\ref{model} we start by recalling the fermionic Fock space formalism and its equivalent formulation in terms of local fermionic modes. This leads to the definition of a categorical model for fermionic circuits: the monoidal category $\cat{LFM}$ of local fermionic modes and maps that either preserve or reverse the parity of a state, with the tensor product of $\mathbb{Z}_2$-graded Hilbert spaces as the monoidal product.

In Section~\ref{components}, we introduce a number of physical gates from which one may build fermionic quantum circuits: the beam splitter, the phase gates, the fermionic swap gate, and the empty and occupied state preparations. After discussing their implementability in the context of fermionic linear optics, we describe our diagrammatic language with its interpretation in $\cat{LFM}$, and show that all the physical gates have simple diagrammatic representations.

In Section~\ref{sec:axioms}, we list the axioms of the fermionic ZW calculus, and prove several derived equations. We introduce short-hand notation for certain composite diagrams (sometimes called the ``spider'' notation in categorical quantum mechanics~\cite[Section~8.2]{coecke2017picturing}), and prove inductive generalisations of the axioms. Then, in Section~\ref{sec:completeness}, we prove our main theorem, that the fermionic ZW calculus is an axiomatisation of $\cat{LFM}$. We achieve this by defining a normal form for diagrams, from which one can easily read the interpretation in $\cat{LFM}$, and showing that any diagram can be rewritten in normal form using the axioms.

As a first application of the diagrammatic language, we calculate in Section~\ref{sec:mach} the statistics of a simple optical circuit, the fermionic Mach-Zehnder interferometer. Finally, in Section~\ref{dualrail}, we give a categorical treatment of the dual-rail encoding --- a standard method in optical quantum computation~\cite{ralphOpticalQuantumComputation2010} which encodes a qubit in two local fermionic modes. We use this to show how the full expressive power of the ZX calculus can be recovered from a restricted set of fermionic gates.

\section{The model}\label{model}

The physical setting of fermions is usually formalised using the \emph{fermionic Fock space}. In this section, we start by recalling this construction to motivate the definition of a \emph{categorical model} of fermionic processes, where the basic unit of information are \emph{local fermionic modes}~(LFMs).

Let us consider a box containing fermions: we want to describe the state space of this many-particle system. Assume that the space of states of a single fermion is given by a Hilbert space $\cal{H}$ of finite dimension $n$. We additionally assume that the particles in the box are \emph{indistinguishable}. This means that permuting the positions of the particles will only change the state of the box by a global phase. The global phase is determined by the \emph{exchange statistics} of the particle under consideration. In the case of fermions these are described by the following antisymmetry relation:
\[ v \otimes w \sim - w \otimes v.\]
The state space of a many-fermions system is then given by the fermionic Fock space:
\begin{equation} \label{eq:fockdef}
	\fock(\cal{H})= \bigoplus_{k=0}^\infty \cal{H} ^{\ftimes k }
\end{equation}
where the direct sum ranges over the number $k$ of particles in the box and $\ftimes$ denotes the antisymmetric tensor product, that is, the quotient of the tensor product obtained by identifying states related by $\sim$; this is canonically equivalent to the subspace of the tensor product whose elements are of the form $\frac{1}{2}(v \otimes w - w \otimes v)$. In particular, $v \ftimes v = 0$, known as the \emph{Pauli exclusion principle}: this means that no pair of particles can be in the same state of $\cal{H}$, and has the consequence that $\cal{H}^{\ftimes k}$ is $0$ for $k > n$, so the infinite direct sum (\ref{eq:fockdef}) reduces to a finite one. We refer to~\cite{blute1994fock} for a discussion of this construction from a logical-computational viewpoint.

This physical setting has a computational counterpart which becomes apparent once we fix an orthonormal basis $X$ for $\cal{H}$, whose elements are called local fermionic modes (LFMs); we write $\cal{H} \simeq \ell^2(X)$ to denote the choice of basis. This induces a basis for $\cal{H}^{\otimes k}$ given by lists of elements of $X$ of size $k$, and a basis of $\cal{H}^{\ftimes k}$ given by the subsets of $X$ of size $k$. Indeed, denoting by $\pset X$ the \emph{powerset} of $X$, we see that
\begin{equation}
	 \fock \ell^2(X) \simeq \ell^2(\pset X);
\end{equation}
as $\pset X \simeq \mathbb{B}^X$, the space $\fock \ell^2(X)$ has basis given by bit-strings of length $n$ denoting the \emph{occupation value} of each LFM:\@ $0$ if there is no particle in the state corresponding to an LFM (\emph{empty} mode), and $1$ if there is (\emph{occupied} mode). We will use the physicists' \emph{ket} notation $\ket{b_1\ldots b_n}$ for the basis state corresponding to such a bit-string.

The construction $\fock$ extends to an endofunctor on the category $\cat{FHilb}$ of finite-dimen\-sional Hilbert spaces and linear maps, in the obvious way. We illustrate its action on morphisms by an example. Suppose $\cal{H}= \mathbb{C}^2$, so that $\fock (\cal{H})= \mathbb{C}^4$, and

\[ f=
\begin{pmatrix}
a & b \\
c & d
\end{pmatrix} : \cal{H} \rightarrow \cal{H}.\]
We draw $f$ as a complete bipartite graph with edges labelled in the complex numbers. To compute $\fock(f)\ket{b_1 b_2}$, we count flows from left to right starting from vertices with $b_i = 1$, and introduce a minus sign each time two edges cross:
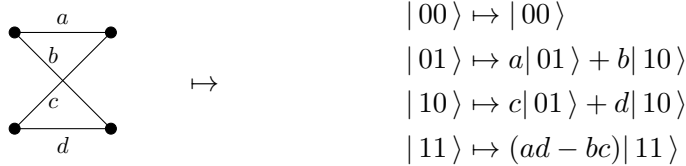
\begin{center}
	\begin{minipage}{0.15\linewidth}
		\begin{tikzpicture}[scale=1.6]
		\node [scale=0.5, draw, circle, fill=black] at (0,0) {};
		\node [scale=0.5, draw, circle, fill=black] at (0,1) {};
		\node [scale=0.5, draw, circle, fill=black] at (1,0) {};
		\node [scale=0.5, draw, circle, fill=black] at (1,1) {};
		\draw (0,0) -- (1,0) node [midway, below] {$d$};
		\draw (0,1) -- (1,0) node [pos=0.4, above] {$b$};;
		\draw (0,0) -- (1,1) node [pos=0.4, below] {$c$};;
		\draw (0,1) -- (1,1) node [midway, above] {$a$};;
		\end{tikzpicture}
	\end{minipage}
	$\mapsto$
	\begin{minipage}{0.55\linewidth}
		\begin{equation*}
		\begin{aligned}
		&\ket{00} \mapsto \ket{00}\\
		&\ket{01} \mapsto a \ket{01} + b \ket{10}\\
		&\ket{10} \mapsto c \ket{01} +d  \ket{10}\\
		&\ket{11} \mapsto (ad - bc) \ket{11}
		\end{aligned}
		\end{equation*}
	\end{minipage}
\end{center}
Note that $\fock (f)$ preserves the number of particles in the box; in fact it can be thought of as applying $f$ independently to each particle. In general, for all linear maps $f$, the map $\fock (f)$ will be \emph{number preserving}. The fermionic processes obtained in this way are known to be classically efficiently simulatable~\cite{terhalClassicalSimulationNoninteractingfermion2002a}; the full power of fermionic quantum computation is obtained by considering a larger class of processes on LFMs.

Physical operations are in fact only restricted to preserve the \emph{parity} of the number of occupied LFMs. Creation and annihilation operators on $\fock(\cal{H})$, which model the introduction or elimination of particles, additionally allow to invert the parity of the number of occupied modes.

This is formalised as follows: the Hilbert space of a system of $n$ LFMs splits as $ \fock (\cal{H}) = H = H_0\oplus H_1$, where $H_0$ is spanned by states where an \emph{even} number of LFMs is occupied, and $H_1$ by states where an \emph{odd} number of LFMs is occupied. Then, any physical operation $f: H_0 \oplus H_1 \to K_0 \oplus K_1$ must either preserve, or invert the parity, that is, either map $H_0$ to $K_0$ and $H_1$ to $K_1$, or map $H_0$ to $K_1$ and $H_1$ to $K_0$. This is called the \emph{parity superselection rule}; see~\cite{banuls2007entanglement, dariano2014feynman} for a discussion.

\begin{rem}Note that Bravyi and Kitaev only consider parity-preserving operations in~\cite{bravyi2002fermionic}, although odd states are allowed.
\end{rem}

These operations assemble into a category, as follows.
\begin{defi}
	A \emph{$\mathbb{Z}_2$-graded Hilbert space} is a complex Hilbert space $H$ decomposed as a direct sum $H_0 \oplus H_1$.

	A \emph{pure map} $f: H \to K$ of $\mathbb{Z}_2$-graded Hilbert spaces is a bounded linear map $f: H \to K$ such that $f(H_0) \subseteq K_0$ and $f(H_1) \subseteq K_1$ (\emph{even} map), or $f(H_0) \subseteq K_1$ and $f(H_1) \subseteq K_0$ (\emph{odd} map).

	Given two $\mathbb{Z}_2$-graded Hilbert spaces $H$, $K$, the tensor product $H \otimes K$ can be decomposed as ${(H \otimes K)}_0 := (H_0 \otimes K_0)\oplus(H_1 \otimes K_1)$, and ${(H \otimes K)}_1 := (H_0 \otimes K_1) \oplus (H_1 \otimes K_0)$. Then, the tensor product (as maps of Hilbert spaces) of a pair of pure maps $f: H \to K$, $f': H' \to K'$ is a pure map $f \otimes f': H \otimes H' \to K \otimes K'$ of $\mathbb{Z}_2$-graded Hilbert spaces. The $\mathbb{Z}_2$-graded Hilbert space $\mathbb{C}\oplus 0$ acts as a unit for the tensor product.

	We write $\graded$ for the symmetric monoidal category of $\mathbb{Z}_2$-graded Hilbert spaces and pure maps, with the tensor product as monoidal product. We write $\graded_0$ for the symmetric monoidal category of $\mathbb{Z}_2$-graded Hilbert spaces and \emph{even} maps.
\end{defi}

\begin{rem}
	The zero map $0: H \to K$ is the only pure map between $\mathbb{Z}_2$-graded Hilbert spaces that is both even and odd.
	Note that the odd maps in $\graded$ do not form a category as the composition of two odd maps is even.
\end{rem}

Because $\fock\cal{H}$ has dimension $2^n$ if $\cal{H}$ has dimension $n$, the fermionic Fock space construction can be seen as a functor $\fock : \cat{FHilb} \rightarrow \cat{LFM}$ where $\cat{LFM}$ is defined below.

\begin{defi}
	We write $\cat{LFM}$ for the full monoidal subcategory of $\graded$ whose objects are $n$-fold tensor products of $B := \mathbb{C}\oplus\mathbb{C}$, for all $n \in \mathbb{N}$. We write $\cat{LFM}_0$ for the subcategory of $\cat{LFM}$ containing only the even maps.
\end{defi}

The generating object $B$ of $\cat{LFM}$ represents a single local fermionic mode and will be our unit of information. $B_0$ is the span of $\ket{0}$ denoting an empty mode, and $B_1$ is the span of $\ket{1}$ denoting an occupied mode. As customary, we write $\ket{b_1\ldots b_n}$ for the basis state $\ket{b_1} \otimes \ldots \otimes \ket{b_n}$ of $B^{\otimes n}$, where $b_i \in \{0,1\}$, for $i = 1,\ldots,n$.  The category $\cat{LFM}$ admits the structure of a dagger compact closed category in the sense of~\cite{selinger2011finite}: each object $B^{\otimes n}$ is self-dual, and the dagger of a pure map $f: B^{\otimes n} \to B^{\otimes k}$ is its adjoint $\dagg{f}: B^{\otimes k} \to B^{\otimes n}$.


\section{Components}\label{components}

Operationally, we are interested in representing circuits built from the following \emph{physical gates}, shown here in diagrammatic form (read from top to bottom), next to their interpretation as maps in $\cat{LFM}$.

\begin{enumerate}
\item The \emph{beam splitter} with parameters $r, t \in \mathbb{C}$, such that $|r|^2 + |t|^2 = 1$:

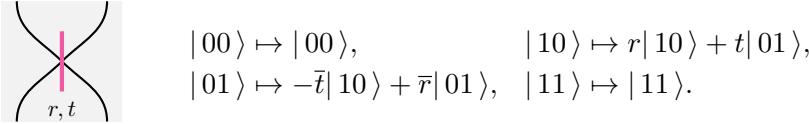
\begin{minipage}{0.15\linewidth}
\input{img/s2_beam-splitter.tex}
\end{minipage}
\begin{minipage}{0.65\linewidth}
	\begin{tabular}{l l}
	$\ket{00} \mapsto \ket{00}$, & $\ket{10} \mapsto r\ket{10} + t\ket{01}$, \\

	$\ket{01} \mapsto -\overline{t}\ket{10} + \overline{r}\ket{01}$, & $\ket{11} \mapsto \ket{11}$.
	\end{tabular}
\end{minipage}

\item The \emph{phase} gate with parameter $\vartheta \in \ropen{0,2\pi}$:

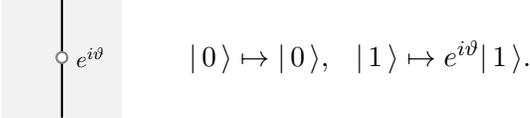
\begin{minipage}{0.15\linewidth}
\begin{tikzpicture}[baseline={([yshift=-.5ex]current bounding box.center)}]
\begin{pgfonlayer}{bg}
	\path[fill, color=gray!10] (-1,-1) rectangle (1,1);
\end{pgfonlayer}
\begin{pgfonlayer}{mid}
	\draw[edge] (0,-1) to (0,1);
	\node[dotwhite] at (0,0) [label=right:{$e^{i\vartheta}$}] {};
\end{pgfonlayer}
\end{tikzpicture}

\end{minipage}
\begin{minipage}{0.65\linewidth}
\begin{tabular}{l l}
	$\ket{0} \mapsto \ket{0}$, & $\ket{1} \mapsto e^{i\vartheta}\ket{1}$.
\end{tabular}
\end{minipage}

\item The \emph{fermionic swap} gate:

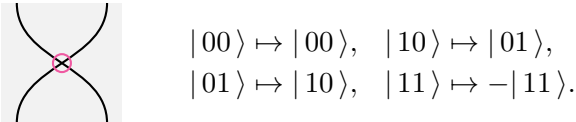
\begin{minipage}{0.15\linewidth}
\input{img/s2_exchange.tex}
\end{minipage}
\begin{minipage}{0.65\linewidth}
\begin{tabular}{l l}
	$\ket{00} \mapsto \ket{00}$, & $\ket{10} \mapsto \ket{01}$, \\
	$\ket{01} \mapsto \ket{10}$, & $\ket{11} \mapsto -\ket{11}$.
\end{tabular}
\end{minipage}

\item \emph{Empty state} and \emph{occupied state} preparation:

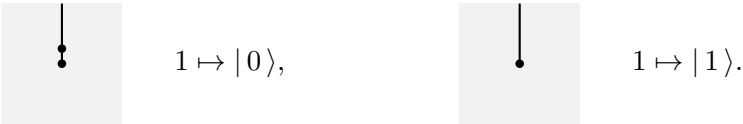
\begin{minipage}{0.15\linewidth}
\begin{tikzpicture}[baseline={([yshift=-.5ex]current bounding box.center)}]
\begin{pgfonlayer}{bg}
	\path[fill, color=gray!10] (-1,-1) rectangle (1,1);
\end{pgfonlayer}
\begin{pgfonlayer}{mid}
	\draw[edge] (0,0) to (0,1);
	\node[dotdark] at (0,0) {};
	\node[dotdark] at (0,0.25) {};
\end{pgfonlayer}
\end{tikzpicture}

\end{minipage}
\begin{minipage}{0.25\linewidth}
	$1 \mapsto \ket{0}$,
\end{minipage}
\begin{minipage}{0.15\linewidth}
\begin{tikzpicture}[baseline={([yshift=-.5ex]current bounding box.center)}]
\begin{pgfonlayer}{bg}
	\path[fill, color=gray!10] (-1,-1) rectangle (1,1);
\end{pgfonlayer}
\begin{pgfonlayer}{mid}
	\draw[edge] (0,0) to (0,1);
	\node[dotdark] at (0,0) {};
\end{pgfonlayer}
\end{tikzpicture}

\end{minipage}
\begin{minipage}{0.25\linewidth}
	$1 \mapsto \ket{1}$.
\end{minipage}
\end{enumerate}

\noindent
All of these are isometries, which makes them, at least in principle, physically implementable gates; see for example~\cite{ji2003electronic} for the description of an electron beam splitter.

Apart from the fermionic swap gate, which exploits the antisymmetry of fermionic particles under exchange, these operations are structurally the same as those used in implementations of linear optical quantum computing (LOQC), such as the Knill-Laflamme-Milburn scheme~\cite{knill2001scheme}, which employ photons, that is, bosonic particles as resources. The two models seem closely related; given the way that the fermionic swap ties the other components together, and that the impossibility for two particles to occupy the same mode --- a constraint for the bosons in LOQC --- is simply a consequence of Pauli exclusion for fermions, it seems likely to us that the logical features of the optical model are a consequence of the features of the fermionic model, rather than the other way around.
We substantiate this analogy in Section~\ref{dualrail} where we discuss the dual-rail encoding, a method used in both bosonic and fermionic linear optics.

In addition to the physical gates, we need the following \emph{structural} components --- the \emph{dualities} and the \emph{swap} --- which allow us to treat all our diagrams as components of a circuit diagram, which can be connected together in an undirected fashion, permuting and transposing their inputs and outputs:

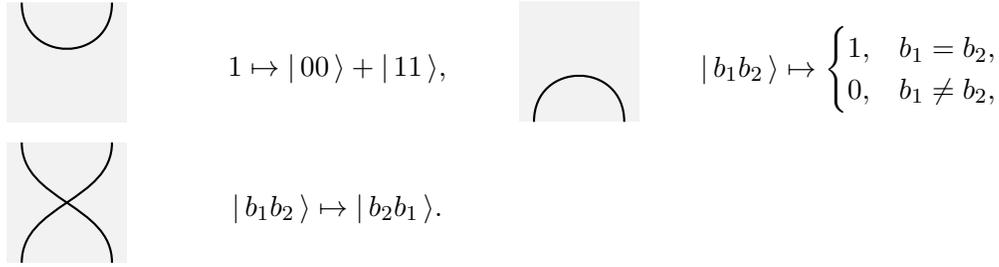
\blankline%

\begin{minipage}{0.13\linewidth}
\begin{tikzpicture}[baseline={([yshift=-.5ex]current bounding box.center)}]
\begin{pgfonlayer}{bg}
	\path[fill, color=gray!10] (-1,-1) rectangle (1,1);
\end{pgfonlayer}
\begin{pgfonlayer}{mid}
	\draw[edge, out=-90, in=-90, looseness=1.75] (-.75,1) to (.75,1);
\end{pgfonlayer}
\end{tikzpicture}

\end{minipage}
\begin{minipage}{0.30\linewidth}
\begin{equation*}
	1 \mapsto \ket{00} + \ket{11},
\end{equation*}
\end{minipage}
\begin{minipage}{0.13\linewidth}
\begin{tikzpicture}[baseline={([yshift=-.5ex]current bounding box.center)}]
\begin{pgfonlayer}{bg}
	\path[fill, color=gray!10] (-1,-1) rectangle (1,1);
\end{pgfonlayer}
\begin{pgfonlayer}{mid}
	\draw[edge, out=90, in=90, looseness=1.75] (-.75,-1) to (.75,-1);
\end{pgfonlayer}
\end{tikzpicture}

\end{minipage}
\begin{minipage}{0.30\linewidth}
\begin{equation*}
	\ket{b_1 b_2} \mapsto \begin{cases} 1, & b_1 = b_2, \\ 0, & b_1 \neq b_2,\end{cases}
\end{equation*}
\end{minipage}

\blankline%

\begin{minipage}{0.13\linewidth}
\input{img/s2_swap.tex}
\end{minipage}
\begin{minipage}{0.30\linewidth}
\begin{equation*}
	\ket{b_1 b_2} \mapsto \ket{b_2 b_1}.
\end{equation*}
\end{minipage}

\blankline%

While the swap, dualities, fermionic swap, and phase gates are going to be basic components of our diagrammatic calculus, we are going to further decompose beam splitters and state preparations. The components so obtained may not correspond to physical operations by themselves, but they have the property that the result of transposing or swapping any of their inputs or outputs only depends on the number of inputs and outputs. This allows us to treat their diagrammatic representations as vertices of an undirected vertex-labelled multigraph: only the overall arity matters. In addition to making calculations simpler, this enables one to implement the calculus in graph rewriting software, such as Quantomatic~\cite{kissinger2015quantomatic}.

The additional components, given here in ``all-output form'' together with their interpretation in $\cat{LFM}$, are the following.

\begin{enumerate}
\item The binary and ternary black vertex:

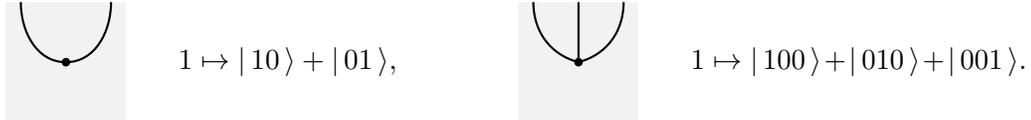
\begin{minipage}{0.15\linewidth}
\input{img/s2_w-binary.tex}
\end{minipage}
\begin{minipage}{0.3\linewidth}
	$1 \mapsto \ket{10} + \ket{01}$,
\end{minipage}
\begin{minipage}{0.15\linewidth}
\input{img/s2_w-ternary.tex}
\end{minipage}
\begin{minipage}{0.3\linewidth}
	$1 \mapsto \ket{100} + \ket{010} + \ket{001}$.
\end{minipage}

\item The binary white vertex with parameter $z \in \mathbb{C}$:

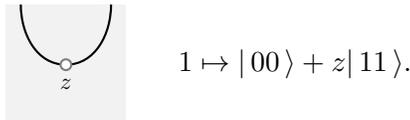
\begin{minipage}{0.15\linewidth}
\input{img/s2_z-binary.tex}
\end{minipage}
\begin{minipage}{0.3\linewidth}
	$1 \mapsto \ket{00} + z\ket{11}$.
\end{minipage}
\end{enumerate}

\begin{rem}
Up to a normalising factor, the interpretations of the binary and ternary vertex are known as EPR state and W state, respectively, in qubit theory~\cite{dur2000three}.
\end{rem}
When we draw black and white vertices with a different partition of inputs and outputs, we assume that a particular partial transposition has been fixed, for example to the left:
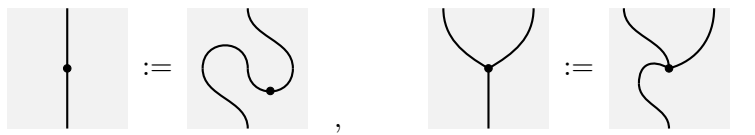
\begin{equation*}
\input{img/s2_w-transposition.tex}
\end{equation*}
Now, a phase gate with parameter $\vartheta$ is simply a binary white vertex with parameter $e^{i\vartheta}$. The beam splitter with parameters $r, t$, and the state preparations can be decomposed as follows:
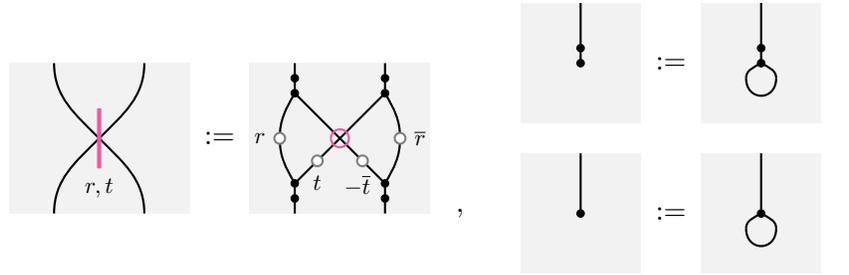
\begin{equation*}
\input{img/s2_gates-decomp.tex}
\end{equation*}

More generally, we can now represent circuits in the image of the Fock space functor --- defined in the previous section --- starting from the complete bipartite graph of the input matrix as in the following example:

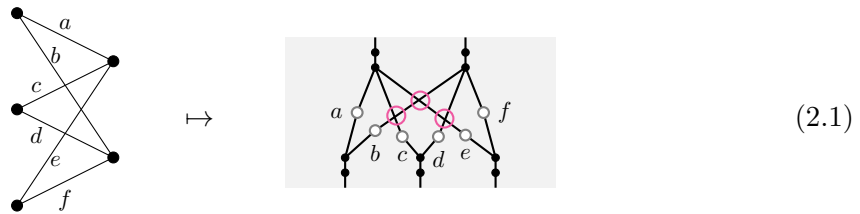
\input{img/s2_fockspace}

Note that $\fock$ maps unitaries to unitaries, which means we can transport unitaries from $\cat{FHilb}$ to $\cat{LFM}$. Most of the physical gates (such as beam splitters, phases and the fermionic swap gate) lie in the image of the Fock space functor.

\begin{rem}
Morphisms $B^{\otimes 2} \to B^{\otimes 2}$ in the image of $\fock$ belong to, but do not exhaust the class of \emph{matchgates}. Valiant showed that circuits built from the proper subclass of next-neighbour matchgates are classically efficiently simulatable~\cite{valiantQuantumComputersThat2001}. These were related to fermionic linear optics in~\cite{terhalClassicalSimulationNoninteractingfermion2002a}. If the next-neighbour restriction is dropped then universal quantum computation is recovered, as shown by Jozsa et al~\cite{jozsaMatchgatesClassicalSimulation2008}; this can be achieved by adding the swap gate to next-neighbour matchgates. In future work, we plan to give a fully diagrammatic characterisation of matchgates and next-neighbour matchgates.
\end{rem}

As we will show in Section~\ref{dualrail}, the full expressiveness of quantum computation is restored by giving access to the \emph{projector} on the even subspace of two LFMs. We introduce a simplified notation for this component, which will play an important role in our axiomatisation:

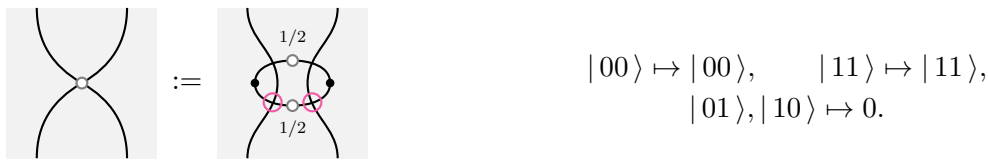
\blankline%

\begin{minipage}{0.4\linewidth}
\input{img/s2_projector.tex}
\end{minipage}
\begin{minipage}{0.55\linewidth}
\begin{equation*}
	\ket{00} \mapsto \ket{00}, \qquad \ket{11} \mapsto \ket{11},
\end{equation*}
\begin{equation*}
	\ket{01}, \ket{10} \mapsto 0.
\end{equation*}
\end{minipage}

\blankline%

\noindent As the notation suggests, this corresponds to the quaternary white vertex of the original ZW calculus. Similarly to the black and white vertices already introduced, its interpretation is symmetric under transposition and swapping of inputs and outputs, so we can freely draw quaternary white vertices with a different partition of inputs and outputs.

\begin{rem}
Our calculus does not include measurements, probabilistic mixing, or any kind of classical control as internal operations. In future work, we hope to extend our axiomatisation to a mixed quantum-classical calculus, in the style of~\cite{coecke2007measurements} (see~\cite[Chapter 8]{coecke2017picturing} for a more recent version), incorporating all these elements.

For now, we can calculate the probability of detecting particles at the output ends of a circuit by closing the circuit with occupied and empty state diagrams; a closed circuit is then interpreted as a map $\mathbb{C} \to \mathbb{C}$, that is, a scalar. This will be the probability amplitude of detecting a particle where we have closed with an occupied state, and not detecting it where we have closed with an empty state.
\end{rem}

To reason rigorously about our diagrammatic calculus, we rely on the theory of PROs~\cite{lack2004composing}, strict monoidal categories that have $\mathbb{N}$ as set of objects, and monoidal product given, on objects, by the sum of natural numbers. Morphisms $n \to m$ in a PRO represent operations with $n$ inputs and $m$ outputs. Given a \emph{monoidal signature}, that is, a set of operations with arities $T := {\{f_i: n_i \to m_i\}}_{i \in I}$, one can generate the free PRO $F[T]$ on $T$, whose operations are free sequential and parallel compositions of the $f_i$, modulo the axioms of monoidal categories. By a classic result of Street and Joyal~\cite[Theorem 1.2]{joyal1991geometry}, this is equivalent to the PRO whose morphisms are obtained by horizontally juxtaposing and vertically plugging string diagrams with the correct arity, one for each generator, then quotienting by planar isotopy of diagrams. Thus, in the remainder, we will not distinguish between the two, identifying diagrams and operations.

\begin{defi}
Let $T$ be the monoidal signature with operations $\mathsf{swap}: 2 \to 2$, $\mathsf{dual}: 0 \to 2$, $\dagg{\mathsf{dual}}: 2 \to 0$, $\mathsf{fswap} : 2 \to 2$, $\mathsf{black}_2: 0 \to 2$, $\mathsf{black}_3: 0 \to 3$, and $\mathsf{white}_z: 0 \to 2$, for all $z \in \mathbb{C}$. The \emph{language} of the fermionic ZW calculus is the free PRO $F[T]$.
\end{defi}
The correspondence with the diagrammatic components we listed earlier should be self-explanatory, and their interpretation induces a monoidal functor $\textit{f}: F[T] \to \cat{LFM}$.

Given a set $E$ of equations between diagrams with the same arity in $F[T]$, we can quotient $F[T]$ by the smallest equivalence relation including $E$ and compatible with composition and monoidal product, to obtain a PRO $F[T/E]$, together with a quotient functor $p_E: F[T] \to F[T/E]$.

\begin{defi} The interpretation $f: F[T] \to \cat{LFM}$ is \emph{universal} if it is a full functor. A set $E$ of equations is \emph{sound} for $f: F[T] \to \cat{LFM}$ if $f$ factors through $p_E: F[T] \to F[T/E]$. A sound set of equations is \emph{complete} for $\cat{LFM}$ if $F[T/E]$ and $\cat{LFM}$ are equivalent monoidal categories.
\end{defi}

In the next section, we will introduce the axioms of the fermionic ZW calculus, in the form of equations between diagrams of $F[T]$; we will then show that they are complete for $\cat{LFM}$.


\section{Axioms and derived equations}\label{sec:axioms}

We divide the set $E$ of axioms into four groups, based on the generators to which they mainly pertain.
\begin{axgrp}\label{ax1} \emph{Structural axioms.}
\begin{equation*}
\input{img/s3_structural1.tex}
\end{equation*}
\begin{equation*}
\input{img/s3_structural2.tex}
\end{equation*}
\begin{equation*}
\input{img/s3_structural3.tex}
\end{equation*}
\begin{equation*}
\input{img/s3_structural4.tex}
\end{equation*}
Together, these axioms imply that the swap and dualities make $F[T/E]$ a compact closed category on a self-dual object. The Kelly-Laplaza coherence theorem~\cite[Theorem 8.2]{kelly1980coherence} then allows us to extend our topological reasoning to the swapping and transposition of wires.
\end{axgrp}

\begin{axgrp}\label{ax2} \emph{Axioms for the fermionic swap.}
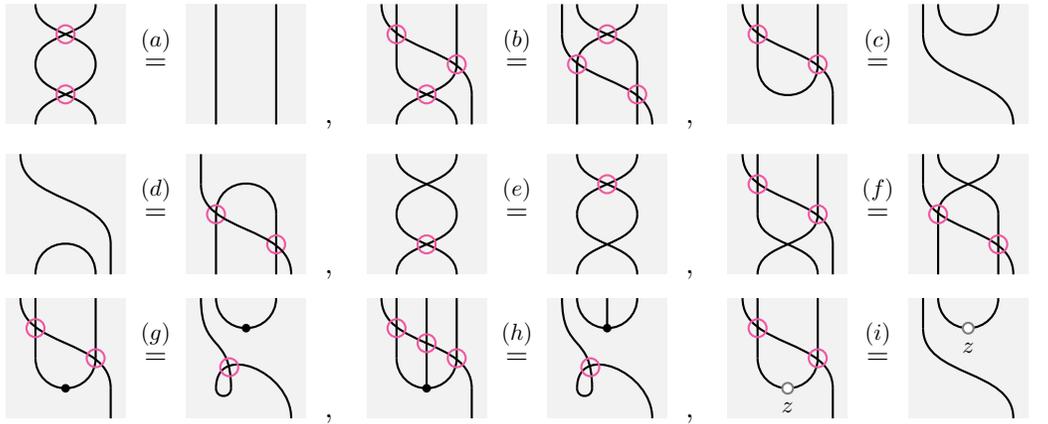
\begin{equation*}
\input{img/s3_fermionic1.tex}
\end{equation*}
\begin{equation*}
\input{img/s3_fermionic2.tex}
\end{equation*}
\begin{equation*}
\input{img/s3_fermionic3.tex}
\end{equation*}
These axioms say that the fermionic swap behaves like a symmetric braiding in $F[T/E]$, except for the fact that sliding the black vertices (that is, the only odd generators) through a wire induces a ``self-crossing'' on it.

Moreover, the axioms on the interplay between the structural and fermionic swaps imply that only the number of fermionic swaps between two wires matters, and not their direction; which, as we will see, also implies that a sequence of two fermionic self-crossings on either side of a wire can be straightened.

Altogether, the result of these axioms is that any diagram containing an \emph{even} number of black vertices can slide past a wire through fermionic swaps with no other effect, while any diagram containing an \emph{odd} number of black vertices can do the same by introducing a self-crossing on the wire. As with the structural axioms, we will make use of this fact implicitly most of the time.
\end{axgrp}

\begin{rem}
Calculi with two different braidings are studied in virtual knot theory; see~\cite{kauffman2012virtual} for an introduction.
\end{rem}

\begin{axgrp}\label{ax3} \emph{Axioms for black vertices.}
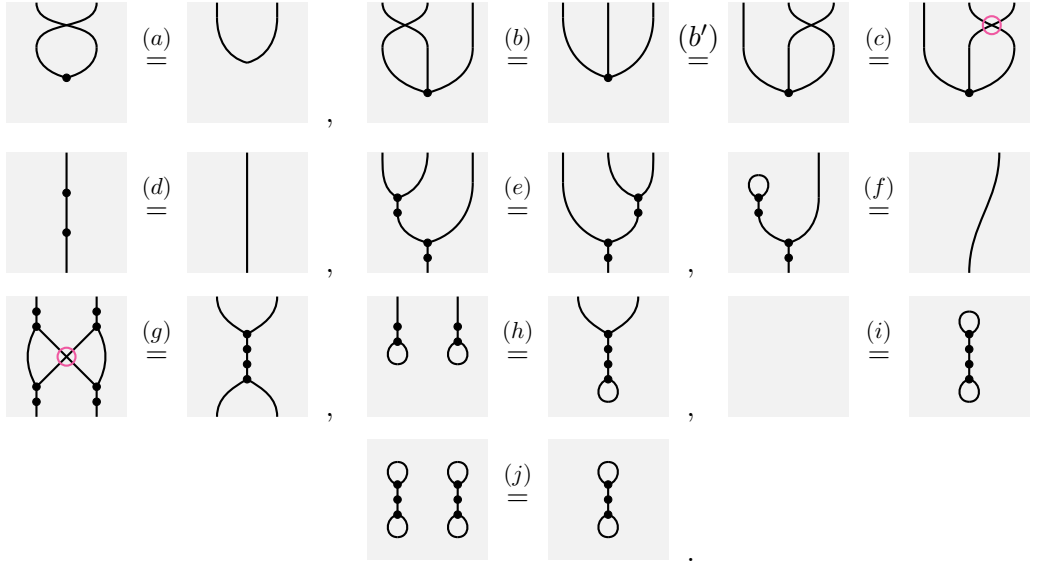
\begin{equation*}
\input{img/s3_black1.tex}
\end{equation*}
\begin{equation*}
\input{img/s3_black2.tex}
\end{equation*}
\begin{equation*}
\input{img/s3_black3.tex}
\end{equation*}
\begin{equation*}
\input{img/s3_black4.tex}
\end{equation*}
These axioms say that the black vertices are symmetric under permutation of wires (which justifies, a posteriori, their arbitrary transposition), and that they can be assembled to form a (co)commutative (co)monoid. This (co)monoid has the property of forming a bialgebra (in fact, a Hopf algebra) with its own transpose. 

In the interpretation, this is the Hopf algebra known as \emph{fermionic line} in the theory of quantum groups~\cite[Example 14.6]{majid2002quantum}, whose comultiplication is given by $\ket{0} \mapsto \ket{00}$, $\ket{1} \mapsto \ket{10} + \ket{01}$. As discussed in~\cite[Section 5.3]{hadzihasanovic2017algebra}, the fermionic line has ``anyonic'' and ``bosonic'' analogues in every countable dimension, with the same self-duality property.

The final axiom says that $0$ times $0$ is $0$; it serves to ensure that there is a unique zero map, rather than an ``even'' and an ``odd'' zero.
\end{axgrp}

\begin{axgrp}\label{ax4} \emph{Axioms for white vertices.}
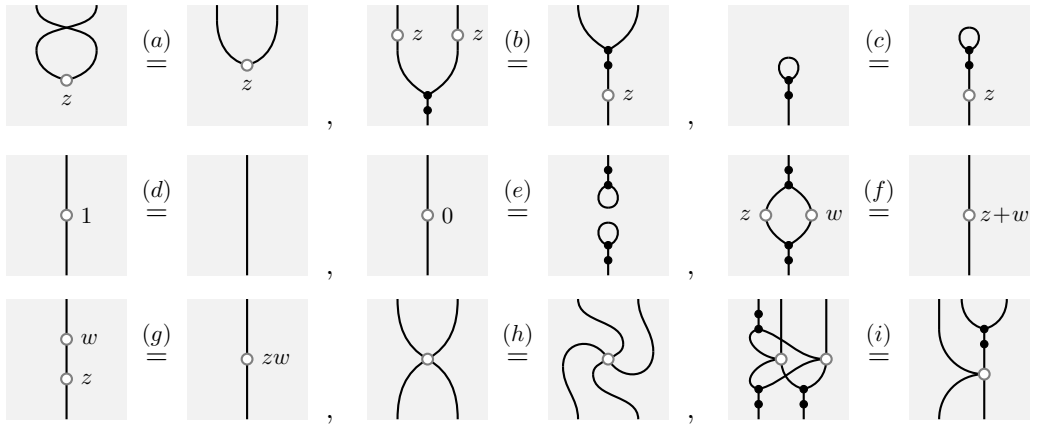
\begin{equation*}
\input{img/s3_white1.tex}
\end{equation*}
\begin{equation*}
\input{img/s3_white2.tex}
\end{equation*}
\begin{equation*}
\input{img/s3_white3.tex}
\end{equation*}
These axioms say that the binary white vertices are homomorphisms for the fermionic line algebra, and that composition and convolution by the algebra correspond to product and sum, respectively, of their complex parameters. Finally, the projector is symmetric under cyclic permutation of its wires, and it determines a kind of mixed action/coaction of the algebra on itself.
\end{axgrp}

\begin{rem}
Because $\cat{LFM}$ is a subcategory of the category $\cat{Qubit}$ of~\cite{lics2018zwzx}, all the axioms of the fermionic ZW calculus are sound for the original ZW calculus. Moreover, adding either the ternary or the unary white vertex from the original ZW calculus to our language would make it universal for $\cat{Qubit}$. We have not yet investigated, however, what axioms would need to be added to $E$ in the extended signature to make it complete.
\end{rem}

We state some useful derived equations, leaving the proofs to the Appendix.

\begin{prop}\label{black-derived}
The following equations hold in $F[T/E]$:
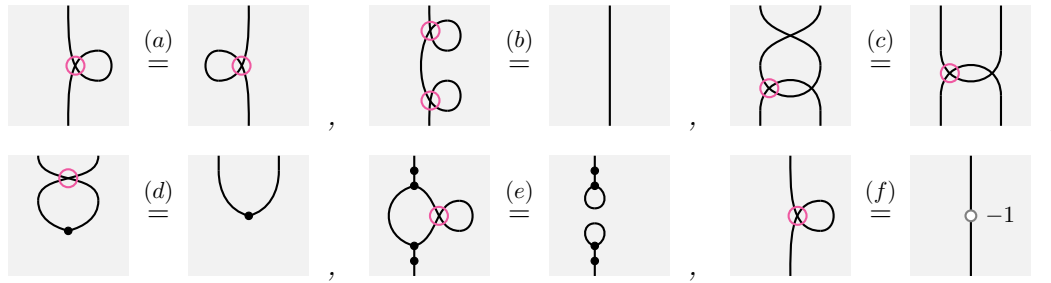
\begin{equation*}
\input{img/s3_fermionic-derived.tex}
\end{equation*}
\begin{equation*}
\input{img/s3_black-derived.tex}
\end{equation*}
\end{prop}

\begin{proof}
	Equation $(a)$ comes from the following manipulation:
	\begin{equation*}
		\input{img/a_leftrightcross.tex}
	\end{equation*}

	\noindent Equation $(b)$ then follows from $(a)$, combined with the equation

	\begin{equation*}
		\input{img/a_whitney.tex}
	\end{equation*}

	\noindent which is a consequence of the fermionic swap axioms by the Whitney trick~\cite[p.\ 484]{kauffman2001knots}. Equation $(c)$ is proved by the following argument:

	\begin{equation*}
		\input{img/a_cz-commute.tex}
	\end{equation*}

	\noindent whereas $(d)$ comes from
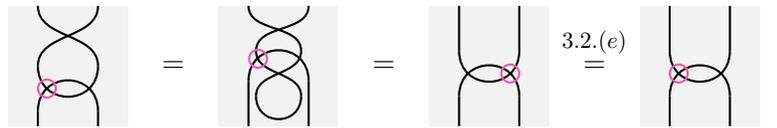
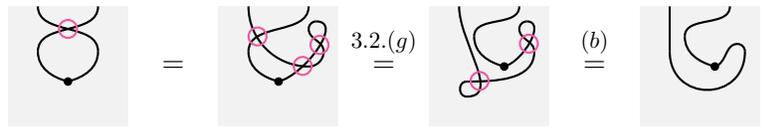
	\begin{equation*}
		\input{img/a_black-swap.tex}
	\end{equation*}

	\noindent finally using the symmetry of the black vertex under the structural swap.

	Equation $(e)$ is proved by the following argument:
	\begin{equation*}
		\input{img/a_hopf1.tex}
	\end{equation*}

	\noindent where we tacitly used Axiom~\ref{ax3}$.(d)$ to introduce or eliminate pairs of binary black vertices in several occasions.

	Finally, for equation $(f)$, start by considering that
	\begin{equation*}
		\input{img/a_minusone.tex}
	\end{equation*}

	\noindent by Axioms~\ref{ax3}$.(e)$ and~\ref{ax4}$.(d)$, this is equal to
	\begin{equation*}
	\input{img/a_minusone2.tex}
	\end{equation*}

	\noindent This completes the proof.
\end{proof}

The following Proposition shows that all the least natural-looking axioms about white vertices in the original ZW calculus become provable equations for (quaternary, rather than ternary) white vertices in the fermionic ZW calculus.

\begin{prop}\label{white-derived}
The following equations hold in $F[T/E]$:
\begin{equation*}
\input{img/s3_white-derived.tex}
\end{equation*}
\begin{equation*}
\input{img/s3_white-derived2.tex}
\end{equation*}
\begin{equation*}
\input{img/s3_white-derived3.tex}
\end{equation*}
\end{prop}

\begin{proof}
	Substituting the definition of the projector, Axiom~\ref{ax4}$.(h)$ becomes the following equation:
	\begin{equation} \label{a_projector}
	\input{img/a_projector.tex}
	\end{equation}
	Equations $(a)$ and $(a')$ are then immediate consequences of Proposition~\ref{black-derived}$.(c)$ and its transposes, applied to the right-hand side of (\ref{a_projector}).

	Equation $(b)$ is also immediate from the definition: because swaps slide through fermionic swaps and vice versa, we can slide one ``circle'' past another to get
	\begin{equation*}
		\input{img/a_white-associativity.tex}
	\end{equation*}

	\noindent For equations $(c)$, $(c')$, and $(c'')$, we use either of the forms in equation (\ref{a_projector}) and slide the binary white vertex through a fermionic swap using Axiom~\ref{ax2}$.(i)$, to move it to a different wire.

	Equation $(d)$ comes from
	\begin{equation*}
		\input{img/a_circle.tex}
	\end{equation*}

	\noindent finally applying Axiom~\ref{ax3}$.(i)$. Then, equation $(e)$ follows from it by
	\begin{equation*}
		\input{img/a_white-trace.tex}
	\end{equation*}

	\noindent In order to prove equation $(f)$, consider first that
	\begin{equation} \label{bw-loop}
	\input{img/a_bw-loop1.tex}
	\end{equation}
	and we can eliminate the circle by equation $(d)$. Then,
	\begin{equation*}
		\input{img/a_bw-loop2.tex}
	\end{equation*}

	\noindent Finally, for equation $(g)$, observe that the projector contains an even number of black vertices, hence it can slide past fermionic swaps with no other effect. Therefore,
	\begin{equation*}
		\input{img/a_uncross.tex}
	\end{equation*}

	\noindent This concludes the proof.
\end{proof}

Together with its invariance under cyclic permutation of wires, the first two equations justify the arbitrary transposition of inputs and outputs of the quaternary white vertex.

Our axioms form a sound and complete set of equations for $\cat{LFM}$, so in principle any equation of diagrams whose interpretations are equal can be derived from them. In practice, however, it is convenient to introduce further short-hand notation, including black vertices with $n$ wires and white vertices with $2n$ wires for all $n \in \mathbb{N}$, and derive inductive equation schemes to use directly in proofs.

\begin{enumerate}
\item \emph{Black vertices.} The nullary and unary black vertices are defined as follows:
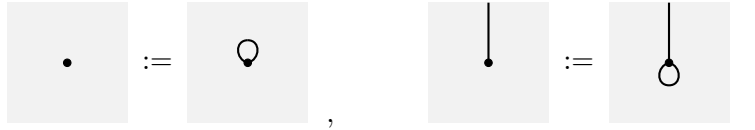
\begin{equation*}
\input{img/s3_black-01.tex}
\end{equation*}
We already have binary and ternary black vertices. For $n > 3$, the $n$-ary black vertex is defined inductively, together with its interpretation in $\cat{LFM}$, by

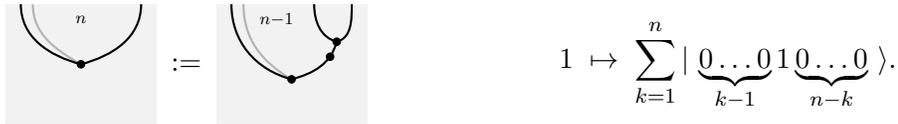
\blankline%

\begin{minipage}{0.40\linewidth}
\input{img/s3_black-n.tex}
\end{minipage}
\begin{minipage}{0.50\linewidth}
\begin{equation*}
	1 \;\mapsto \;\sum_{k=1}^n \ket{\underbrace{0\ldots 0}_{k-1}1\underbrace{0\ldots 0}_{n-k}\,}.
\end{equation*}
\end{minipage}

\blankline%

\noindent Here, and in what follows, lighter wires and vertices indicate the repetition of a pattern for a number of times, which may or may not be specified. This is similar to the the way that ``$\ldots$'' is often used, and can be formalised using $!$-boxes in pattern graphs, as developed in~\cite{kissinger2014pattern}. 

\item[(2)] \emph{White vertices.} The nullary white vertex with parameter $z \in \mathbb{C}$ is defined by
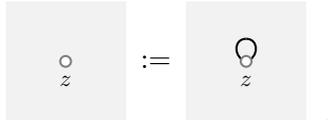
\begin{equation*}
\input{img/s3_white-0.tex}
\end{equation*}
We already have binary white vertices with parameters. For $n > 1$, the $2n$-ary white vertex with parameter $z \in \mathbb{C}$ is defined inductively, together with its interpretation in $\cat{LFM}$, by

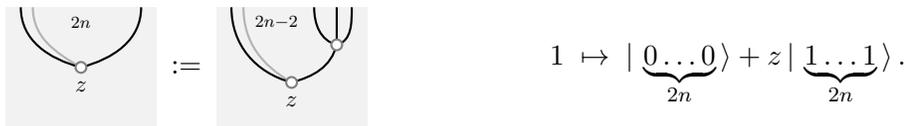
\blankline%

\begin{minipage}{0.40\linewidth}
\input{img/s3_white-2n.tex}
\end{minipage}
\begin{minipage}{0.50\linewidth}
\begin{equation*}
	1 \; \mapsto \; \ket{\underbrace{0\ldots 0}_{2n}} + z \, \ket{\underbrace{1\ldots 1}_{2n}}\,.
\end{equation*}
\end{minipage}

\blankline%

\end{enumerate}
We state some basic properties of black and white vertices. Both are symmetric under permutation of wires, which allows us to write vertices with different numbers of inputs and outputs, transposing some of them with no ambiguity. Most importantly, they satisfy certain ``fusion'' equations, as shown on the second and third line. All black vertices correspond to odd maps, while white vertices correspond to even maps, as reflected in their sliding through fermionic swaps, on the fourth line; finally, black vertices are unaffected by fermionic swaps of their wires, whereas the sign of the white vertex parameter is flipped.
\begin{prop}\label{prop:spider}
The following equations hold in $F[T/E]$ for black and white vertices of any arity:
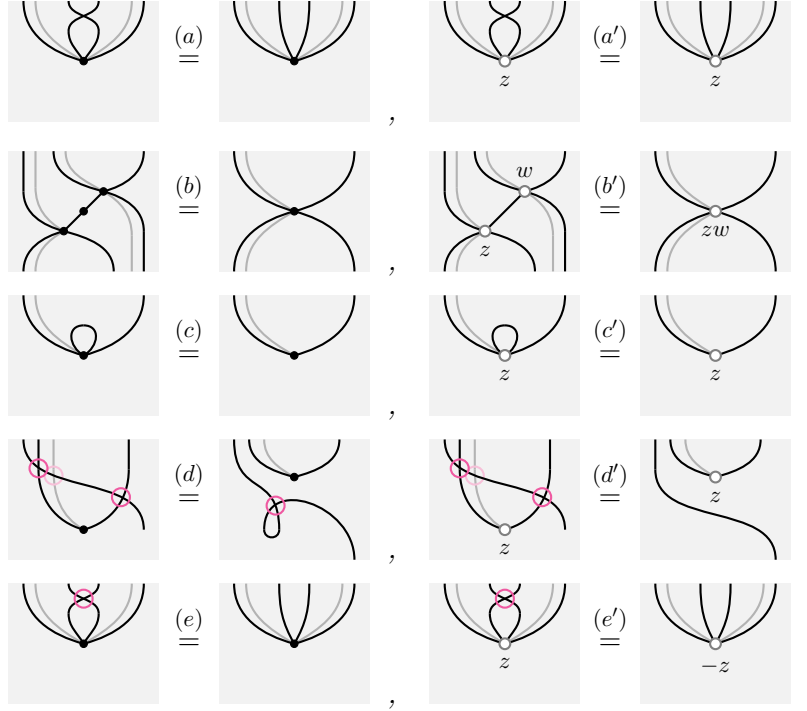
\begin{equation*}
\input{img/s3_spider-swap.tex}
\end{equation*}
\begin{equation*}
\input{img/s3_spider-fusion.tex}
\end{equation*}
\begin{equation*}
\input{img/s3_spider-trace.tex}
\end{equation*}
\begin{equation*}
\input{img/s3_spider-slide.tex}
\end{equation*}
\begin{equation*}
\input{img/s3_spider-fermionic.tex}
\end{equation*}
\end{prop}

\begin{proof}
	All the equations are proved by induction on the arity of the vertices involved.

	For equation $(a)$, let $n$ be the number of outputs of the black vertex. For $n = 0,1$ there is nothing to prove, and for $n = 2,3$ these are Axioms~\ref{ax3}$.(a)$, $(b)$, and $(b')$. For $n > 3$, if the two swapped wires are the rightmost ones, the equation follows immediately from the ternary case; otherwise, use Axiom~\ref{ax3}$.(e)$ on the three rightmost wires, and apply the inductive hypothesis.

	For equation $(a')$, let $2n$ be the number of outputs of the white vertex. For $n=0$ there is nothing to prove, and $n = 1$ is Axiom~\ref{ax4}$.(a)$. For $n > 1$, observe that by Proposition~\ref{white-derived}$.(c)$, $(c')$ and $(c'')$, we can always move the binary vertex with parameter $z$ to a wire which is not swapped. The case $n = 2$ then follows from the combination of Axiom~\ref{ax4}$.(h)$ with Proposition~\ref{white-derived}$.(a)$ and $(a')$. For $n > 2$, if the swapped wires are among the three rightmost ones, the equation follows from the case $n = 2$; otherwise, use Proposition~\ref{white-derived}$.(b)$ (with some wires transposed) on the two rightmost quaternary white vertices, and apply the inductive hypothesis.

	Equations $(a)$ and $(a')$ justify the unambiguous writing of $n$-ary vertices with inputs as well as outputs in equations $(b)$ and $(b')$, and the latter will follow from the all-output case. In equation $(b)$, let $n, m > 0$ be the arities of the leftmost and rightmost vertex, respectively. If $n = 1$, the equation follows from Axiom~\ref{ax3}$.(f)$, and if $n = 2$ from Axiom~\ref{ax3}$.(d)$. Suppose $n > 2$. Then, if $m = 1$, the equation follows from Axiom~\ref{ax3}$.(f)$, and if $m = 2$ from Axiom~\ref{ax3}$.(d)$. All other cases are just immediate from the definition. Equation $(b')$ also follows from the definition, together with Proposition~\ref{white-derived}$.(c)$, $(c')$ and $(c'')$ in order to move the vertex with parameter $w$ to the wire where the vertex with parameter $z$ is, and Axiom~\ref{ax4}$.(g)$ to multiply the two.

	In equation $(c)$, let $n > 1$ be the arity of the black vertex in the left-hand side. If $n = 2, 3$ the equation is true by definition. If $n > 3$, by equation $(a)$, we can assume the two wires plugged into each other are the two rightmost ones; the equation then follows from Axiom~\ref{ax3}$.(f)$. For equation $(c')$, let $2n > 1$ be the arity of the white vertex in the left-hand side. If $n = 1$, the equation is true by definition, and if $n = 2$ it follows from Proposition~\ref{white-derived}$.(d)$, together with Proposition~\ref{white-derived}$.(c)$, $(c')$ and $(c'')$ to move the vertex with parameter $z$ out of the way. For $n > 2$, by equation $(a')$, we can assume the two wires plugged into each other are the two rightmost ones, and the equation follows from the case $n = 2$.

	Equation $(d)$ is a consequence of Axioms~\ref{ax2}$.(g)$ and $(h)$, together with Proposition~\ref{black-derived}$.(b)$ to eliminate pairs of self-crossings. Equation $(d')$ is a consequence of Axiom~\ref{ax2}$.(i)$ together with the definition of the quaternary white vertex.

Equation $(e)$ follows from equation $(a)$ by Axiom~\ref{ax3}$.(c)$ and Proposition~\ref{black-derived}$.(d)$. For equation $(e')$, let $2n > 1$ be the arity of the white vertex. The case $n = 1$ is a consequence of Proposition~\ref{black-derived}$.(f)$, and $n = 2$ follows from the following argument:

	\begin{equation*}
		\input{img/a_white-negate.tex}
	\end{equation*}
applied to the definition of the quaternary white vertex, as in the right-hand side of (\ref{a_projector}). All other cases follow from this one, by symmetry.
\end{proof}

Several other equations, both axioms and derived, admit inductive generalisations; we list them in the following Proposition.

\begin{prop}\label{prop:inductive}
The following equations hold in $F[T/E]$ for black and white vertices of any compatibile arities:
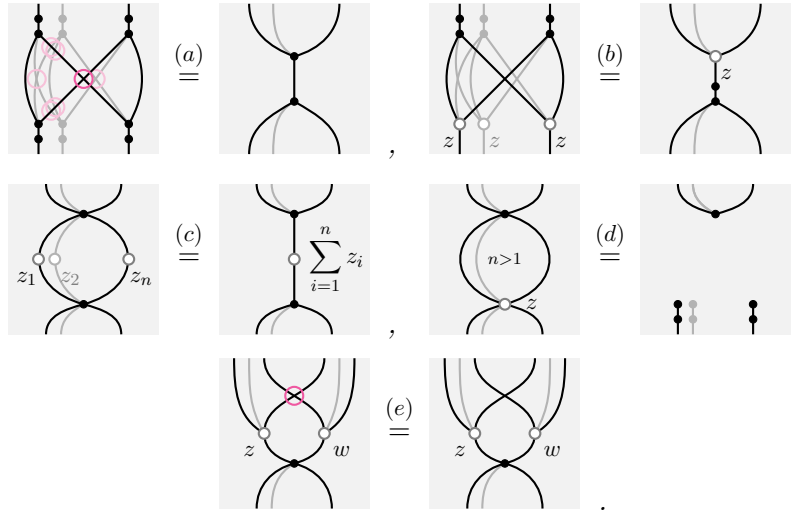
\begin{equation*}
\input{img/s3_inductive1.tex}
\end{equation*}
\begin{equation*}
\input{img/s3_inductive2.tex}
\end{equation*}
\begin{equation*}
\input{img/s3_inductive3.tex}
\end{equation*}
\end{prop}

\begin{proof}
	In equation $(a)$, let $n$ be the number of inputs, and $m$ the number of outputs of the diagrams. The case $n = m = 0$ is Axiom~\ref{ax3}$.(i)$, and when either $n$ or $m = 1$, the equation follows from Axiom~\ref{ax3}$.(d)$. The cases $n = 0, m > 1$ and $m = 0, n > 1$ are simple inductive generalisations of Axiom~\ref{ax3}$.(h)$. Finally, the case $n = m = 2$ is Axiom~\ref{ax3}$.(g)$, and from there we can proceed by double induction on $n$ and $m$, using Proposition~\ref{prop:spider}$.(b)$.

	In equation $(b)$, let $n$ be the number of inputs, and $2m-1$, for $m > 0$, the number of outputs. Suppose first that $m = 1$. The case $n = 0$ is Axiom~\ref{ax4}$.(c)$, the case $n = 1$ follows from Axiom~\ref{ax3}$.(d)$, and the case $n = 2$ is Axiom~\ref{ax4}$.(b)$; then, for $n > 2$, it is a simple induction starting for the latter. In the case $n = 0$ and $m = 2$,
	\begin{equation*}
		\input{img/a_bialgebra1.tex}
	\end{equation*}

	\noindent by Proposition~\ref{prop:spider}$.(b')$ and $(c')$, the latter is equal to
	\begin{equation*}
		\input{img/a_bialgebra2.tex}
	\end{equation*}

	\noindent The cases $n = 0$, $m > 2$ are simple inductive generalisations of this one. All cases with $n = 1$ follow from Axiom~\ref{ax3}$.(d)$, and the case $n = m = 2$ is Axiom~\ref{ax4}$.(i)$. For $n, m > 2$, proceed by double induction, using Proposition~\ref{prop:spider}$.(b)$ and $(b')$.

	For equation $(c)$, by Proposition~\ref{prop:spider}$.(b)$ it suffices to prove
	\begin{equation*}
		\input{img/a_sum.tex}
	\end{equation*}

	\noindent which for $n = 0$ is Axiom~\ref{ax4}$.(e)$, for $n = 1$ follows from Axiom~\ref{ax3}$.(d)$, for $n = 2$ is Axiom~\ref{ax3}$.(f)$, and for $n > 2$ is a simple inductive generalisation of the latter.

	Similarly, for equation $(d)$, it suffices, by Proposition~\ref{prop:spider}$.(b)$ and $(b')$, to prove
	\begin{equation*}
		\input{img/a_loop.tex}
	\end{equation*}

	\noindent when $m = 0, 1, 2$. If $m = 2$, and $n = 2$, this is Proposition~\ref{white-derived}$.(f)$, and for $n > 2$ we can proceed by induction, as follows:
	\begin{equation*}
		\input{img/a_loop2.tex}
	\end{equation*}

	\noindent The case $m = 0$, for arbitrary $n > 1$, follows from this one, by
	\begin{equation*}
		\input{img/a_loop3.tex}
	\end{equation*}

	\noindent and similarly for the case $m = 1$, where necessarily $n > 2$, by
	\begin{equation*}
		\input{img/a_loop4.tex}
	\end{equation*}

	\noindent Finally, equation $(e)$ is an immediate generalisation of Proposition~\ref{white-derived}$.(g)$, using Proposition~\ref{prop:spider}$.(b)$ and $(b')$.
\end{proof}

\begin{rem}
Some of these inductive schemes subsume several axioms at once: for example, Proposition~\ref{prop:inductive}$.(a)$ has Axioms~\ref{ax3}$.(g)$,~\ref{ax3}$.(h)$, and~\ref{ax3}$.(i)$ as special cases, and Proposition~\ref{prop:inductive}$.(b)$ has Axioms~\ref{ax4}$.(b)$,~\ref{ax4}$.(c)$, and~\ref{ax4}$.(i)$ as special cases.

It is possible to perform diagram rewriting directly with such equation schemes, using for example the theory of pattern graphs developed in~\cite{kissinger2014pattern}.
\end{rem}

In the next section, we will use these equations to prove that our axioms are complete for $\cat{LFM}$.


\section{Normal form and completeness}\label{sec:completeness}

We will prove completeness in three stages:
\begin{enumerate}
	\item First, we will associate to any state $v: \mathbb{C} \to B^{\otimes n}$ of $\cat{LFM}$ a diagram $g(v): 0 \to n$ in $F[T]$ such that $f(g(v)) = v$. Because both categories are compact closed, and the dualities of $\cat{LFM}$ are in the image of $f$, this assignment can be extended to any map of $\cat{LFM}$, proving universality of our interpretation. We will say that a diagram is in \emph{normal form} if it is of the form $g(v)$ for some $v$.
	\item Then, we will show that any composite of diagrams in normal form can be rewritten in normal form using the equations in $E$, proving that $g$ determines a monoidal functor from $\cat{LFM}$ to $F[T/E]$.
	\item Finally, we will show that all the generators of $F[T]$ can be rewritten in normal form using the equations in $E$, proving that $g$ and $f$ are two sides of a monoidal equivalence between $F[T/E]$ and $\cat{LFM}$.
\end{enumerate}

\begin{thm}[Universality]\label{thm:universality} 
The functor $f: F[T] \to \cat{LFM}$ is full.
\end{thm}
\begin{proof}
Write an arbitrary state $v: \mathbb{C} \to B^{\otimes n}$ in the form $1 \mapsto \sum_{i=1}^m z_i \ket{b_{i1}\ldots b_{in}}$, where $z_i \neq 0$ for all $i$, and no pair of $n$-tuples $(b_{i1},\ldots,b_{in})$ is equal; we can fix an ordering (for example, lexicographic) on $n$-tuples of bits to make this unique. Then, define
\begin{equation} \label{normalform}
\input{img/s4_normalform.tex}
\end{equation}
where, for $i = 1,\ldots, m$ and $j= 1,\ldots, n$, the dotted wire connecting the $i$-th white vertex to the $j$-th output is present if and only if $b_{ij} = 1$. The definition is only ambiguous if $v = 0$, in which case we arbitrarily pick one of the two forms; they will be equal in $F[T/E]$ by Axiom~\ref{ax3}$.(j)$.

Because for all summands of an odd (respectively, even) state $v$, we have $b_{ij} = 1$ for an odd (respectively, even) number of bits, the white vertices in $g(v)$ have an odd (respectively, even) number of outputs. The two distinct forms of $g(v)$ for odd and even states ensure that only white vertices with an even arity appear.

It can then be checked that $f(g(v)) = v$, which, by our earlier remark, suffices to prove the statement.
\end{proof}

\begin{defi}
A string diagram of $F[T]$ is in \emph{normal form} if it is $g(v)$ for some state $v$ of $\cat{LFM}$.
It is in \emph{pre-normal form} if it has one of the two forms in (\ref{normalform}), where the following are also allowed:
\begin{itemize}
	\item the white vertices can be in an arbitrary order;
	\item two or more white vertices may be connected to the exact same outputs;
	\item $z_i$ may be 0 for some $i$.
\end{itemize}
\end{defi}
\begin{exa}
The following diagrams are in pre-normal form, but only the second one is in normal form:
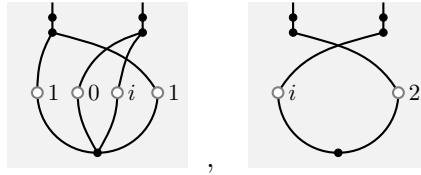
\begin{equation*}
\input{img/s4_prenormal.tex}
\end{equation*}
Both are interpreted in $\cat{LFM}$ as the state $1 \mapsto i\ket{01} + 2\ket{10}$.
\end{exa}

The completeness proof closely follows that of the qubit ZW calculus~\cite[Section 5.2]{hadzihasanovic2017algebra}. The one proof that is significantly different is the following. We take the liberty of ``zooming in'' on a certain portion of a diagram, which may require some reshuffling of vertices, using swapping or transposition of wires, with the implicit understanding that this can always be reversed later.

\begin{lem}[Negation]\label{lem:negation}
The composition of one output of a diagram in pre-normal form with a binary black vertex can be rewritten in pre-normal form, and that has the effect of ``complementing'' the connections of the output to white vertices: that is, locally,
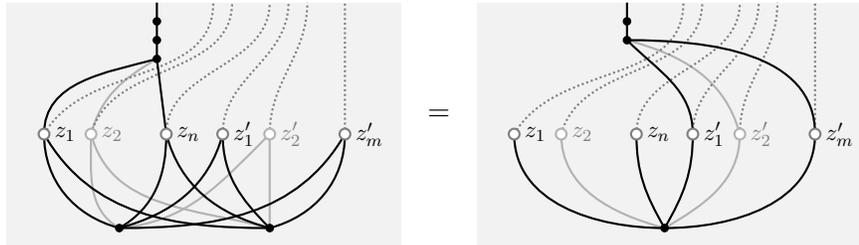
\begin{equation*}
\input{img/s4_negation.tex}
\end{equation*}
\end{lem}
\begin{rem}
In the picture, the dotted wires can stand for a multitude of wires. The version where the original diagram is odd, rather than even, is obtained by composing again both sides with a binary black vertex and using Axiom~\ref{ax3}$.(d)$.
\end{rem}
\begin{proof}
Using the ``fusion rules'' Proposition~\ref{prop:spider}$.(b)$ and $(b')$, we rewrite the left-hand side as
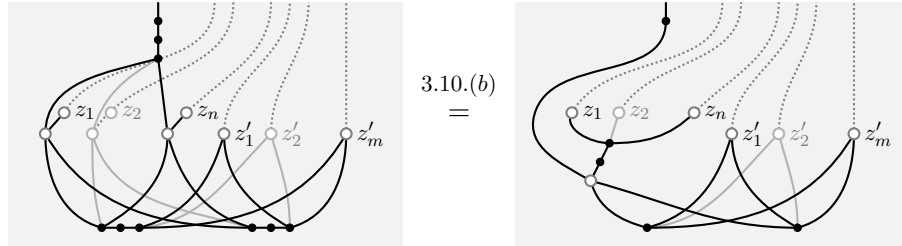
\begin{equation*}
\input{img/s4_negationproof1.tex}
\end{equation*}
By definition of the quaternary white vertex, this is equal to
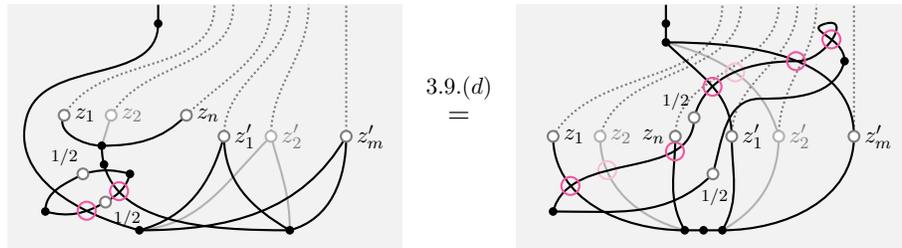
\begin{equation*}
\input{img/s4_negationproof2.tex}
\end{equation*}
where we made implicit use of some symmetry properties of vertices. Now, fusing black vertices, and using Proposition~\ref{prop:spider}$.(d)$ and $(d')$ to move the closed loop to the outside of the main diagram, we see that this is equal to
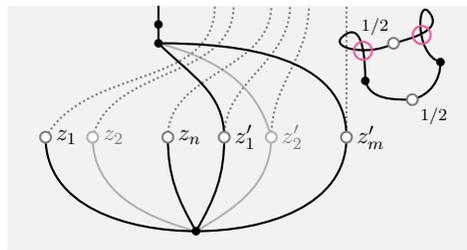
\begin{equation*}
\input{img/s4_negationproof3.tex}
\end{equation*}
and we can conclude by Proposition~\ref{black-derived}$.(b)$ and~\ref{white-derived}$.(e)$.
\end{proof}

In the following, and later statements, ``plugging one output of a diagram into another'' means a post-composition with $\dagg{\mathsf{dual}}: 2 \to 0$, possibly after some swapping of wires.

\begin{lem}[Trace]\label{lem:trace}
The plugging of two outputs of a diagram in pre-normal form into each other can be rewritten in pre-normal form.
\end{lem}
\begin{proof}
We consider the case of an odd diagram; the even case is completely analogous. Focus on the two relevant outputs, and subdivide the white vertices into four groups, based on their being connected to both outputs, to one of them, or neither of them:
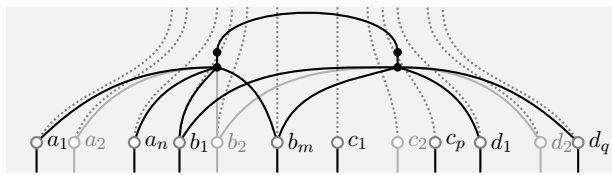
\begin{equation*}
\input{img/s4_trace.tex}
\end{equation*}
Again, the dotted wires can stand for a multitude of wires. Using the negation lemma on the rightmost output, this becomes
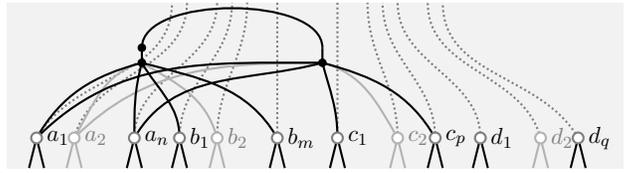
\begin{equation*}
\input{img/s4_trace2.tex}
\end{equation*}
where there are now two black vertices (not pictured) at the bottom, one leading to all the left-hand inputs, and one leading to all the right-hand inputs of the white vertices.

After fusing the pictured black vertices, the leftmost $n$ white vertices have two wires connecting them to the same black vertex. This means that Proposition~\ref{prop:inductive}$.(d)$ is applicable, leaving us with
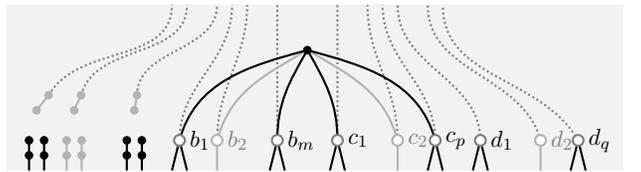
\begin{equation*}
\input{img/s4_trace3.tex}
\end{equation*}
The leftmost black vertices can be fused with any black vertices they are connected to, or eliminated with Axiom~\ref{ax3}$.(i)$ if there is none, which rids us of the leftmost $n$ white vertices. Now, we can apply the negation lemma again, to find that the remaining portion of the diagram is equal to
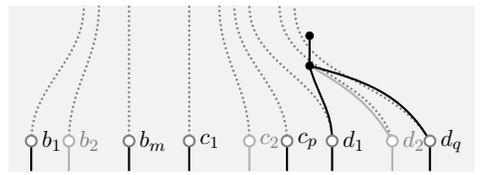
\begin{equation*}
\input{img/s4_trace4.tex}
\end{equation*}
which, focussing on the rightmost part, is equal, by Proposition~\ref{prop:inductive}$.(a)$, to
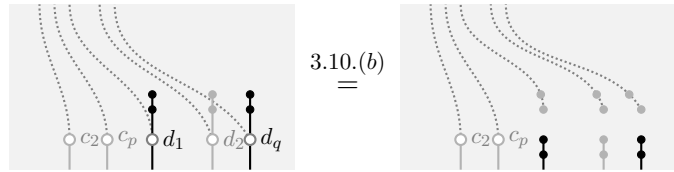
\begin{equation*}
\input{img/s4_trace5.tex}
\end{equation*}
This rids us of the rightmost $q$ white vertices, and leaves us with a diagram in pre-normal form.
\end{proof}

The nullary black vertex is interpreted as the scalar 0; the following lemma shows that it acts as an ``absorbing element'' for diagrams in pre-normal form.

\begin{lem}[Absorption]\label{lem:absorb}
For all diagrams in pre-normal form, the following equation holds in $F[T/E]$:
\begin{equation} \label{eq:absorbed}
\input{img/s4_absorption.tex}
\end{equation}
\end{lem}
\begin{proof}
If the diagram is even, expanding the nullary black vertex, we can treat it as an additional output of the diagram, with no connections to the white vertices, composed with a unary black vertex; then the proof is the same as~\cite[Lemma 5.25]{hadzihasanovic2017algebra}.

Suppose the diagram is odd. If it has at least one output wire, we can freely introduce two binary black vertices on it; applying the negation lemma once, we obtain a negated even diagram, to which the first part of the proof can be applied. Another application of the negation lemma, followed by Axiom~\ref{ax3}$.(j)$, produces the desired equation. If the diagram has no outputs, it necessarily consists of a single nullary black vertex, and the statement follows immediately from Axiom~\ref{ax3}$.(j)$.
\end{proof}

\begin{lem}[Functoriality of the normal form]\label{thm:functoriality}
Any composition of two diagrams in pre-normal form can be rewritten in pre-normal form.
\end{lem}
\begin{proof}
We can factorise any composition of diagrams in pre-normal form as a tensor product followed by a sequence of ``self-pluggings''; thus, by the trace lemma, it suffices to prove that a tensor product --- diagrammatically, the juxtaposition of two diagrams in pre-normal form --- can be rewritten in pre-normal form.

Suppose first that the two diagrams are both even. Then, we can create a pair of unary black vertices connected by a wire by Axiom~\ref{ax3}$.(i)$, and treat them as additional outputs, one for each diagram. Applying the negation lemma on both sides, we obtain
\begin{equation*}
\input{img/s4_functoriality.tex}
\end{equation*}
which is the plugging of two outputs connected to all the white vertices of their respective diagrams. The only case in which this still leaves the two diagrams disconnected is when one of the diagrams has no white vertices, that is, it looks like the right-hand side of equation (\ref{eq:absorbed}). In this case, by the absorption lemma, we can use its isolated black vertex to ``absorb'' the other diagram, which produces a diagram in pre-normal form.

So, suppose that $n, m > 0$. Focussing on the two outputs,
\begin{equation} \label{eq:wbialgebra}
\input{img/s4_functoriality2.tex}
\end{equation}
then, we can use Proposition~\ref{prop:inductive}$.(b)$ on each of the white vertices: for example, on the leftmost one, it leads to
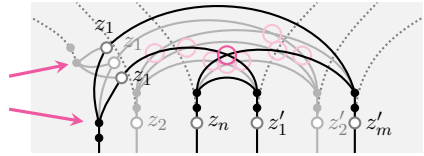
\begin{equation*}
\input{img/s4_functoriality3.tex}
\end{equation*}
Each of the outgoing wires leads to a black vertex, so we can push the vertices indicated by arrows to the outside, and fuse them. Repeating this operation, we can push all black vertices to the outside, which leaves us with a tangle of $n\cdot m$ wires connecting white vertices, one for each pair $(z_i, z'_j)$, where $i = 1, \ldots, n$, and $j = 1, \ldots, m$.

This tangle is made of fermionic swaps; however, each of the white vertices has at least one wire connecting it to a black vertex, which means that Proposition~\ref{prop:inductive}$.(e)$ is applicable: this allows us to turn all the fermionic swaps into structural swaps. Finally, we can fuse all pairs of white vertices connected by a wire. After some rearranging, this leaves us with the pre-normal form diagram
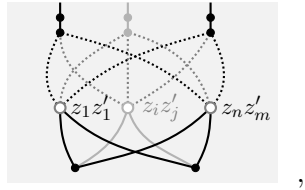
\begin{equation*}
\input{img/s4_functoriality4.tex}
\end{equation*}
where the white vertex with parameter $z_i z'_j$ is connected both to the outputs to which the original vertex with parameter $z_i$ was connected, and to the outputs to which the original vertex with parameter $z'_j$ was connected, for $i = 1,\ldots, n$, and $j = 1, \ldots, m$.

Now, suppose one diagram is odd, or they both are odd. If the odd diagrams have at least one output wire, we can introduce a pair of black vertices on it, and apply the negation lemma to produce negated even diagrams. We can then apply the first part of the proof to obtain a diagram in pre-normal form negated once or twice, then apply the negation lemma again to conclude. If one of the odd diagrams has no outputs, it necessarily consists of a single nullary black vertex, and we can conclude with an application of the absorption lemma.
\end{proof}

\begin{lem}\label{lem:prenormal}
Any diagram in pre-normal form can be rewritten in normal form.
\end{lem}
\begin{proof}
First of all, by the symmetry property of vertices, we can always reshuffle the white vertices of a diagram in pre-normal form to make them follow our preferred ordering of $n$-tuples of bits.

Suppose the diagram is odd, and two white vertices are connected to the same outputs. The relevant portion of the diagram looks like
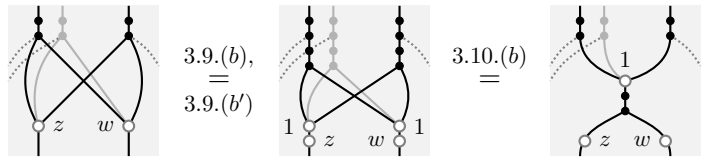
\begin{equation*}
\input{img/s4_prenormalproof1.tex}
\end{equation*}
The two input wires both lead to the bottom black vertex; zooming in on that, we find
\begin{equation} \label{eq:prenormal_proof}
\input{img/s4_prenormalproof2.tex}
\end{equation}
which has a single white vertex, connected to the same outputs, replacing the two initial ones.

Now, suppose that there is a white vertex with parameter 0. The relevant part of the diagram is
\begin{equation*}
\input{img/s4_prenormalproof3.tex}
\end{equation*}
where we can fuse black vertices, which simply eliminates the white vertex.

If the diagram is even, and has at least one output wire, we can introduce a pair of binary black vertices, apply the negation lemma once to produce a negated odd diagram, reduce that to normal form, and apply the negation lemma again; it is easy to see that negation turns diagrams in normal form into diagrams in normal form, \emph{modulo} a reshuffling of white vertices.

If the diagram has no output wires, then it is of the form
\begin{equation*}
\input{img/s4_prenormalproof4.tex}
\end{equation*}
where the right-hand side is in normal form. This concludes the proof.
\end{proof}

\begin{thm}[Completeness]
The functor $f_E: F[T/E] \to \cat{LFM}$ induced by the soundness of $E$ for the interpretation $f: F[T] \to \cat{LFM}$ is a monoidal equivalence.
\end{thm}
\begin{proof}
By the combination of the previous two lemmas, it suffices to show that all the generators (with all wires transposed to output wires) can be rewritten in normal form. For the ternary and binary black vertices,
\begin{equation*}
\input{img/s4_genrewrite1.tex}
\end{equation*}
For the binary white vertex with parameter $z \in \mathbb{C}$,
\begin{equation*}
\input{img/s4_genrewrite2.tex}
\end{equation*}
By Axiom~\ref{ax4}$.(d)$, the rewriting of dualities in normal form follows as a special case of the binary white vertex with parameter $1$.

For the fermionic swap, we use the fact that we know how to rewrite the tensor product of two dualities in normal form:
\begin{equation*}
\input{img/s4_genrewrite4.tex}
\end{equation*}
which, by Proposition~\ref{prop:spider}$.(e')$ and~\ref{prop:inductive}$.(e)$, is equal to
\begin{equation*}
\input{img/s4_genrewrite5.tex}
\end{equation*}
The case of the structural swap is similar, and easier. This concludes the proof.
\end{proof}

\begin{rem} We can make this an equivalence of dagger compact closed categories, by defining the dagger of a morphism in $F[T/E]$, represented by a diagram in $F[T]$, to be the vertical reflection of that diagram, with parameters $z \in \mathbb{C}$ of white vertices turned into their complex conjugates $\overline{z}$. For example,
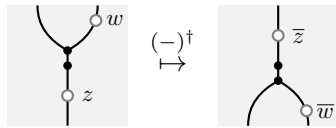
\begin{equation*}
\input{img/s4_dagger.tex}
\end{equation*}
\end{rem}

\begin{rem}
The only properties of complex numbers that were used in the proof are that they form a commutative ring, and that they contain an element $z$ such that $z + z = 1$ (namely, $1/2$). Thus, we can replace $\mathbb{C}$ with any commutative ring $R$ that has the latter property (for example, $\mathbb{Z}_{2n+1}$, for each $n \in \mathbb{N}$), and obtain a similar completeness result for ``LFMs with coefficients in $R$''.

Moreover, for any such ring, instead of introducing binary white vertices with arbitrary parameters $r \in R$, we can introduce one binary white vertex for each element of a family of generators of $R$, together with one axiom for each relation that they satisfy. Then, in the normal form, instead of having a white vertex labelled $r \in R$ at each end of the bottom black spider(s), we will need to have some expression of $r$ by sums and products of generators, encoded by composition and convolution by the fermionic line algebra. 

The completeness proof still goes through: we work with diagrams in pre-normal form, where terms in a sum of products of generators are decomposed into different legs of the bottom spider(s), until the very end, then proceed as in Lemma~\ref{lem:prenormal}, but stop before performing the steps in equation (\ref{eq:prenormal_proof}). 

For example, in the complex case, it may be convenient to have separate phase gates, that is, white vertices with parameter $e^{i\vartheta}$, for $\vartheta \in \ropen{0,2\pi}$, and ``resistor'' gates, with real parameter $r > 0$.
\end{rem}

\begin{rem}
	It is customary to describe the fermionic behaviour of a multi-particle system in terms of a pair of operators $\dagg{a}$ (creation) and $a$ (annihilation) that satisfy the anti-commutation relation $a\dagg{a} = 1 - \dagg{a}a$; see for example~\cite[Chapter 27]{woit2017quantum}. In our language, these operators can be defined as
	\begin{equation*}
		\input{img/s3_anticommutation1.tex}
	\end{equation*}

	\noindent We can see the anti-commutation relation as subsumed by the axioms in the following way: pulling back the linear structure of $\cat{LFM}$ to $F[T/E]$ through the equivalence, we have
	\begin{equation} \label{antipode}
	\input{img/s3_anticommutation2.tex}
	\end{equation}
	from which we obtain
	\begin{equation*}
		\input{img/s3_anticommutation3.tex}
	\end{equation*}

	\noindent which can be read as the equation $a\dagg{a} = 1 - \dagg{a}a$.
\end{rem}


\section{An application: the Mach-Zehnder interferometer}\label{sec:mach}

The Mach-Zehnder interferometer is a classic quantum optical setup (see for example~\cite{scully1997quantum}), which, despite its simplicity, can demonstrate interesting features of quantum mechanics, as in the Elitzur-Vaidman bomb tester experiment~\cite{elitzur1993quantum}. The theoretical setup can be straightforwardly imported into FQC, with the same statistics as long as single-particle experiments are concerned; an electronic analogue of the Mach-Zehnder interferometer has also been realised in practice~\cite{ji2003electronic}.

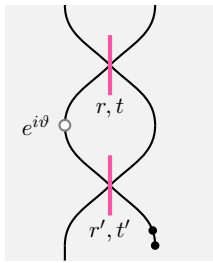
\blankline%

\begin{minipage}{0.25\linewidth}
\input{img/s5_machzehnder.tex}
\end{minipage}
\begin{minipage}{0.7\linewidth}
With the graphical notation introduced in Section~\ref{components}, the experimental setup is represented by the diagram on the left, where $r, t, r', t' \in \mathbb{C}$ and $\vartheta \in \ropen{0,2\pi}$ are parameters subject to $|r|^2 + |t|^2 = |r'|^2 + |t'|^2 = 1$. In practice, it would also include ``mirrors'', or beam splitters with $|r| = 1$, which we omit in the picture, instead taking the liberty of bending wires at will.
\end{minipage}

\blankline%

\noindent As a first application of the fermionic ZW calculus, we show how this circuit diagram can be simplified in just a few steps using our axioms, in such a way that its statistics become immediately readable from the diagram.

In our language, the diagram becomes
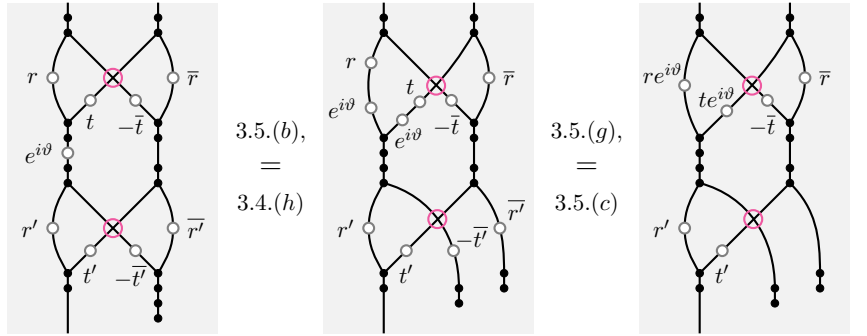
\begin{equation*}
\input{img/s5_simplify1.tex}
\end{equation*}
which, sliding the leftmost empty state past the fermionic swap, and using Axiom~\ref{ax3}$.(f)$ twice, becomes
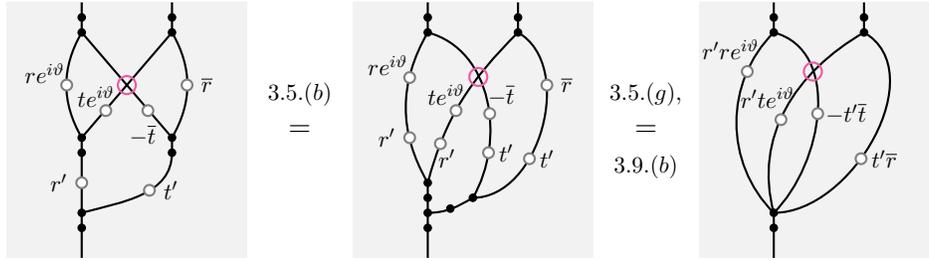
\begin{equation*}
\input{img/s5_simplify2.tex}
\end{equation*}
Finally, using the fermionic swap symmetry of black vertices (Proposition~\ref{prop:spider}$.(e)$), together with Proposition~\ref{prop:inductive}$.(c)$, this simplifies to
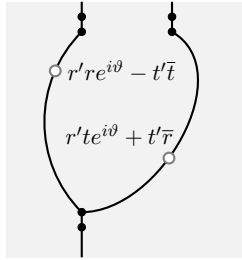
\begin{equation*}
\input{img/s5_simplify3.tex}
\end{equation*}
If we input one particle, after fusing the bottom black vertices, we obtain a diagram in normal form, whose interpretation in $\cat{LFM}$ we can readily deduce:
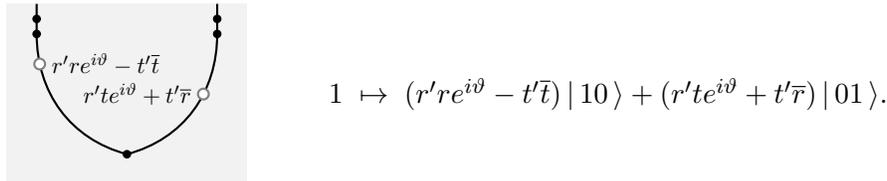
\begin{equation*}
\input{img/s5_finalstate.tex}
\end{equation*}
So, the probability of detecting the particle at the left-hand output is $|r're^{i\vartheta} - t'\overline{t}|^2$, and the probability of detecting the particle at the right-hand output is $|r'te^{i\vartheta} + t'\overline{r}|^2$. If the beam splitters are symmetric, that is, $r = r' = \frac{1}{\sqrt{2}}$, and $t = t' = \frac{i}{\sqrt{2}}$, the probability amplitudes become
\begin{equation*}
	\frac{1}{2}(e^{i\vartheta} - 1) = e^{i(\frac{\vartheta + \pi}{2})}\sin\vartheta, \quad \quad \frac{i}{2}(e^{i\vartheta} + 1) = e^{i(\frac{\vartheta + \pi}{2})}\cos\vartheta,
\end{equation*}
leading to probabilities $\sin^2\vartheta$ of detecting the particle at the left-hand output, and $\cos^2\vartheta$ of detecting it at the right-hand output.

Arguably, given that this particular example involves at most binary gates, a matrix calculation would not have been considerably harder. On the other hand, the result appears here as the outcome of a short sequence of intuitive, algebraically motivated local steps, rather than the unexplained product of a large matrix multiplication. We expect the advantage to become clearer when implementations of rewrite strategies in graph rewriting software are used to simplify larger circuits.


\section{Dual-rail encoding and universal quantum computation}\label{dualrail}

Dual-rail encodings are a standard method in optical computing to perform universal quantum computation~\cite{ralphOpticalQuantumComputation2010, knill2001scheme}. Indeed, as the Hadamard gate has mixed parity, it cannot be obtained on a single fermionic particle in $\mathbf{LFM}$: the dual-rail encoding allows to recover this gate in the even sector of two fermionic modes. This means we will consider processes on pairs of wires: two modes are used as computational units. The \emph{logical qubit} is encoded in the even sector of $B \otimes B$, that is, the subspace generated by $\ket{00}$ and $\ket{11}$.

We will define a functorial mapping from pairs of LFMs to qubits realizing this encoding. In fact the dual-rail encoding arises from a more general relationship between the category $\graded_0$ and $\cat{Hilb}$. Consider the functor:
\begin{equation}
	\begin{aligned}
	{(-)}_0: \, &\graded_0 \rightarrow \cat{Hilb}\\
	& H \mapsto H_0\\
	& (f: H \rightarrow K) \mapsto (f_0: H_0 \rightarrow K_0)
	\end{aligned}
\end{equation}
Note that this functor can be extended to the whole of $\graded$ by sending odd maps to zero. More importantly, ${(-)}_0$ is monoidal and comonoidal. Indeed, for two objects $H,K$ in $\graded_0$ we have
\[{(H\otimes K)}_0 = (H_0 \otimes K_0) \oplus (H_1 \otimes K_1),\]
and the universal property of the direct sum induces a natural map
\[ p_{HK}: {(H\otimes K)}_0 \rightarrow H_0 \otimes K_0.\]
We can represent this natural transformation and its adjoint diagrammatically as follows:
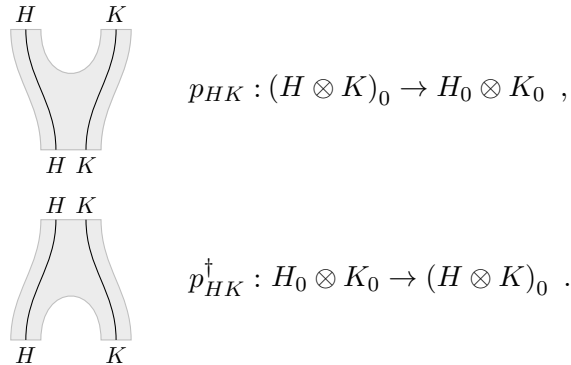
\begin{center}
	\begin{minipage}{0.15\linewidth}
		\begin{tikzpicture}
		\begin{pgfonlayer}{bg}
		\draw[fill, color=gray!15] (0,0) to[out=-90,in=90] (0.5,-2) -- (1.5,-2) to[out=90,in=-90] (2,0)--(1.5,0) to[out=-90,in=-90,looseness=2.5] (0.5,0)--cycle;
		\draw[color=gray!50] (0,0) to[out=-90,in=90] (0.5,-2) -- (1.5,-2) to[out=90,in=-90] (2,0)--(1.5,0) to[out=-90,in=-90,looseness=2.5] (0.5,0)--cycle;
		\end{pgfonlayer}
		\begin{pgfonlayer}{mid}
		\node at (0.25,0.25) {$H$};
		\node at (1.75,0.25) {$K$};
		\node at (0.75,-2.25) {$H$};
		\node at (1.25,-2.25) {$K$};
		\draw (0.25,0) to[out=-90, in= 90] (0.75,-2);
		\draw (1.75,0) to[out=-90, in= 90] (1.25,-2);
		\end{pgfonlayer}
		\end{tikzpicture}
	\end{minipage}
	$p_{HK} :$ 
	\begin{minipage}{0.5\linewidth}
		$ {(H \otimes K)}_0 \rightarrow H_0 \otimes K_0$\enspace,
	\end{minipage}
\end{center}

\begin{center}
	\begin{minipage}{0.15\linewidth}
		\begin{tikzpicture}
		\begin{pgfonlayer}{bg}
		\path[fill, color=gray!15] (0,0) to[out=90,in=-90] (0.5,2) -- (1.5,2) to[out=-90,in=90]  (2,0)--(1.5,0) to[out=90,in=90,looseness=2.5]  (0.5,0)--cycle;
		\draw[color=gray!50] (0,0) to[out=90,in=-90] (0.5,2) -- (1.5,2) to[out=-90,in=90]  (2,0)--(1.5,0) to[out=90,in=90,looseness=2.5]  (0.5,0)--cycle;
		\end{pgfonlayer}
		\begin{pgfonlayer}{mid}
		\node at (0.25,-0.25) {$H$};
		\node at (1.75,-0.25) {$K$};
		\node at (0.75,2.25) {$H$};
		\node at (1.25,2.25) {$K$};
		\draw (0.25,0) to[out=90, in= -90] (0.75,2);
		\draw (1.75,0) to[out=90, in= -90] (1.25,2);
		\end{pgfonlayer}
		\end{tikzpicture}
	\end{minipage}
	$p_{HK}^\dagger :$ 
	\begin{minipage}{0.5\linewidth}
		$ H_0 \otimes K_0 \rightarrow {(H \otimes K)}_0$\enspace.
	\end{minipage}
\end{center}

Naturality is a topological feature of this representation, that is, given $f: H \rightarrow H'$ and $g:K \rightarrow K'$, we have
\input{img/s6_naturality}
and similarly for $p_{HK}^\dagger$.
These natural transformations satisfy the (co)associativity conditions: 
\input{img/s6_associativity}
We also have unit and counit morphisms for ${(-)}_0$:
\begin{center}
	\begin{minipage}{0.15\linewidth}
		\begin{tikzpicture}
		\begin{pgfonlayer}{bg}
		\path[fill, color=gray!15] (0.5,2) -- (1.5,2) to[out=-90,in=-90, looseness=2.5]  (0.5,2);
		\draw[color=gray!50] (0.5,2) -- (1.5,2) to[out=-90,in=-90, looseness=2.5]  (0.5,2);
		\end{pgfonlayer}
		\end{tikzpicture}
	\end{minipage}
	$\epsilon :$ 
	\begin{minipage}{0.5\linewidth}
		$\complex \xrightarrow{\sim} {(\complex)}_0$\enspace,
	\end{minipage}
\end{center}
\begin{center}
	\begin{minipage}{0.135\linewidth}
		\begin{tikzpicture}
		\begin{pgfonlayer}{bg}
		\path[fill, color=gray!15] (0.5,-2) -- (1.5,-2) to[out=90,in=90, looseness=2.5]  (0.5,-2);
		\draw[color=gray!50] (0.5,-2) -- (1.5,-2) to[out=90,in=90, looseness=2.5]  (0.5,-2);
		\end{pgfonlayer}
		\end{tikzpicture}
	\end{minipage}
	$\epsilon^\dagger :$ 
	\begin{minipage}{0.5\linewidth}
		${(\complex)}_0 \xrightarrow{\sim} \complex$\enspace.
	\end{minipage}
\end{center}
Which are inverses of each other (as linear maps they are just the identity on $\complex$), and satisfy the (co)unitality conditions: 

\input{img/s6_unitality}

\noindent Also, it holds that $p_{HK} \circ p^\dagger_{HK} = id_{H_0 \otimes K_0}$:

\input{img/s6_strong}
yet in general
\input{img/s6_hole}
that is, ``holes'' cannot be eliminated. In fact, the LHS of (\ref{holes}) is a projector by (\ref{strong}), that is, it is part of a measurement and can be implemented using postselection.

Now, the dual-rail encoding is obtained by restricting ${(-)}_0$ to $\cat{LFM}_0^2$ --- the full subcategory of $\cat{LFM}_0$ with objects are generated under finite tensor products by $B \otimes B$.

\begin{defi}
	For each morphism $D$ in $\cat{LFM}_0^2$, the dual-rail encoding of $D$ is the following morphism in $\cat{Qubit}$:
	\input{img/s6_encoding}
\end{defi}

We now aim to find a minimal gate set which gives universal quantum computation under the dual-rail encoding. First of all note that we can obtain a Hadamard gate using a beam splitter and two binary black vertices (also known as $X$-gates):

\input{img/s6_hadamard}

\begin{rem}
	Note that the two instances of an $X$-gate are used in order to perform the Hadamard on the even sector of two LFMs: without those, the beam splitter would give a Hadamard gate on the odd sector.
\end{rem}

\noindent Then, $Z$-phases with angle $\vartheta$ are obtained as follows:

\input{img/s6_phases}

Fermionic linear optical computing using only beam splitters and phases is efficiently classically simulatable~\cite{terhalClassicalSimulationNoninteractingfermion2002a}.
As mentioned in Section~\ref{components}, the swap gate allows to recover universal quantum computation, but it is not physically implementable because linear optical setups are restricted to two dimensions.
The even projector, also known as ``parity gate'' in the literature, has been discussed as a powerful resource for fermionic linear optical computing~\cite{ionicioiuEntanglingSpinsMeasuring2007a}. It is known that this gate can be performed using charge measurements on fermions and that this allows universal quantum computation~\cite{beenakkerChargeDetectionEnables2004a}.
We now recover this universality result in our framework: we will see that even projectors on the physical qubits allow to form ``holes'' as in (\ref{holes}) at the logical level. Our notation reveals the topological nature of universal quantum computation with fermions.

\begin{rem}
	We have used a planar notation for the dual-rail encoding to emphasize the two-dimensional restriction of optical quantum computation. This is only a syntactic choice: we could have used \emph{cobordisms} instead to denote our functors, but it would put emphasis on the fact that $\graded_0$ is a modular category~\cite{bartlettExtended3dimensionalBordism2014}, which is an overreach here.
\end{rem}

First note that the green (or $Z$) spider of the ZX calculus is obtained as follows:
\input{img/s6_spider}
The gates in (\ref{hadamard}) and (\ref{phases}) together with (\ref{spider}) generate the ZX calculus, as it appears in~\cite{coecke2017picturing} or in its complete version~\cite{lics2018zwzx}. Therefore, we can now express any ZX diagram through the dual rail-encoding. For example, using the $Z$ spider and the Hadamard gate (\ref{hadamard}), we can now represent a CZ gate:
\input{img/s6_czgate}
Composing Hadamard gates (\ref{hadamard}), phases (\ref{phases}) and $Z$ spiders, we may form holes. However, these can be performed at the physical level using even projectors:

\input{img/s6_killholes}
Thus we can state the following.

\begin{prop}
	Any qubit unitary can be obtained using the dual-rail encoding from fermionic circuits only involving beam splitters, phases, even projectors, $X$-gates and creation-annihilation pairs. In particular, no swap or fermionic swap gates are needed as basic gates.
\end{prop}
\begin{proof}
	Hadamard, phases and CZ form a universal gate set and they can all be obtained as shown above. Composing these will yield diagrams containing holes (and at most two wires between each hole) which we can eliminate using (\ref{killholes}).
\end{proof}

\begin{rem}
	We have only focused on the expressivity aspects of the dual-rail encoding. In future work we would like to understand the rewriting theory of this encoding, which would allow to recover all the ZX rules from logical axioms for the encoding together with the underlying rules of the fermionic ZW calculus.
\end{rem}


\section{Conclusions and outlook}

In this paper, we introduced a string-diagrammatic language for circuits of local fermionic modes, together with equations that axiomatise their theory of extensional equality: that is, two diagrams represent the same linear map of local fermionic modes if and only if they are equal modulo the equations. We believe that these fermionic circuits are to the ZW calculus what Clifford circuits~\cite{backens2014zx} are to the ZX calculus: not the largest family of circuits that can technically be represented, but the one whose basic gates have simple, natural representations in terms of the language's components.

There are still several open questions and directions on the ``syntactic'' side. We do not know whether all our axioms are independent, nor have we looked at rewrite strategies, or ways of orienting the equations, beyond the goal of proving completeness. Additional rewrite strategies could be deduced from the rules of the ZX calculus using the dual-rail encoding to compare the two calculi. There is, then, the question of variants and extensions: we have mentioned a potential extension to mixed-state processes, via a mixed quantum-classical calculus in the style of~\cite[Chapter 8]{coecke2017picturing}; moreover, both universality and completeness are open problems for anyonic and bosonic generalisations of the fermionic ZW calculus, in the style of~\cite[Section 5.3]{hadzihasanovic2017algebra}.

The greatest challenge, however, is finding ``real-world'' applications for the calculus. With the Mach-Zehnder interferometer, we have only given a toy example, perhaps useful for pedagogical purposes, but we have not even attempted to link our work to current research on algorithms or complexity in FQC\@. The first version of a ZW calculus was introduced in order to tackle open problems in the classification of multipartite entanglement~\cite{coecke2010compositional}: as a first step, the fermionic ZW calculus, with its strong topological flavour, involving braidings and a single type of ternary vertices, may be a better testing ground for this approach.


\bibliographystyle{alpha}
\bibliography{main}

\end{document}